\documentclass[sigconf,nonacm]{acmart}

\usepackage{balance}
\usepackage{times,amsmath,epsfig, array,tabularx}
\usepackage{mathrsfs}
\usepackage{subfigure}
\usepackage[linesnumbered,ruled, noend]{algorithm2e}
\usepackage{url}
\usepackage{float}
\usepackage{pifont}

\usepackage{tikz}
\usetikzlibrary{shadings}
\newcommand{\Mt}[2]{ \tikz[baseline=(X.base)] {\node[circle, draw=black, minimum height=6pt, inner sep=1pt, fill=#1](X){#2};}}

\graphicspath{{Figure/}}

\newtheorem{definition}{Definition}[section]
\newtheorem{lemma}{Lemma}
\newtheorem{theorem}{Theorem}[section]

\begin{document}

\title{XMiner: Efficient Directed Subgraph Matching with Pattern Reduction}

\author{
       Pingpeng Yuan, Yujiang Wang, Tianyu Ma, Siyuan He}
\affiliation{%
  \institution{School of Computer Science \& Technology,\\
  Huazhong University of Science and Technology}  
  \city{Wuhan}
  \country{China}
}
\author{Ling Liu}
\affiliation{%
  \institution{College of Computing,\\ Georgia Institute of Technology}  
  \city{Atlanta}
  \country{USA}
}

\begin{abstract}
Graph pattern matching, one of the fundamental graph mining problems,  aims to extract structural patterns of interest from an input graph. 
The state-of-the-art  graph matching algorithms and systems are mainly designed for undirected graphs. Directed graph matching is more complex than undirected graph matching because the edge direction must be taken into account before the exploration of each directed edge. 
Thus, the technologies (e.g. storage, exploiting symmetry for graph matching) for undirected graph matching may not be fully applicable to directed graphs. For example, the redundancy implied in directed graph pattern can not be detected using the symmetry breaking for undirected pattern graph. 
Here, we present XMiner for efficient directed graph pattern matching whose core idea is 'pattern reduction'. It first analyzes the relationship between constraints implied in a pattern digraph. Then it reduces the pattern graph into a simplified form by finding a minimum constraint cover. Finally, XMiner generates an execution plan and follows it to extract matchings of the pattern graph. So, XMiner works on simplified pattern graph and avoids much data access and redundant computation throughout the matching process. 
Our experimental results show that XMiner outperforms state-of the-art stand-alone graph matching systems,
and scales to complex graph pattern matching tasks on larger graphs.
\end{abstract}
\maketitle
\section{Introduction}
\label{sec:intro}
Now, data are huge and highly connected. Analysis of them to discover inherent association is becoming even more important. So, many graph processing systems or applications are developed for the tasks, which can be roughly classified into two kinds: iterative computation and graph analysis \cite{YuanICDCS2019}. 
Graph analytics involves graph mining, which includes subgraph matching/enumeration, subgraph finding, subgraph mining, and graph clustering \cite{G-Miner18}.  The graph matching tasks mainly involve exploring the substructures within the input graph and need to search huge space if data graph is large. Thus, graph matching problems are generally computationally intensive and difficult to process in an efficient way. Many graph matching approaches have been proposed \cite{SunSIGMOD20,Kim2021,GIProblem,GIPNAS2015}.  Since they are single-threaded programs, they can not process complex patterns or large data graphs in reasonable time.  
So, RStream \cite{RStream18}, Fractal \cite{Fractal},
G-Miner \cite{GMiner}, GraphPi \cite{GraphPi}, Pangolin \cite{Pangolin}, Peregrine \cite{Peregrine}, and AutoMine \cite{AutoMine} etc. provide a parallel or distributed framework to  efficiently support graph mining tasks. The graph processing systems exhaustively search a data graph for the target subgraphs (e.g., matchings of a specific pattern) in a step-by-step way. As these subgraphs are explored and are checked using filters, the un-matched subgraphs get pruned and the matchings are returned. In this way, the systems prune the search space. While such exploration frameworks are general enough to compute different graph mining use cases, 
state-of-the-art graph matching approaches face some issues as described next (For the reason of distinction, we call the objects of pattern graph and data graph as vertices and nodes, respectively, call the links of pattern graph and data graph as edges and arcs, respectively.):

\textbf{Few Symmetry}. Many real graphs are directed graphs. Directed graphs have more complex structural patterns and less symmetry than undirected graphs. So, Matching a direct subgraph pattern is more challenging. 
it is more unlikely for directed pattern graphs to have symmetry structures. For instance, two rectangles in Figure \ref{fig:rectangle} are asymmetry while their undirected versions are symmetry. So, some symmetry technologies developed in undirected graph matching approaches \cite{GraphPi,Peregrine} can not efficiently improve the performance of directed graph matching. 

\textbf{Redundant Data Access}. In graph matching problem, the domain of edges or vertices in the given pattern graph is the data graph. So, they may share many matchings when the data graph is big. Thus, there exist many redundant data access operations which are expensive if data graph is stored in slow storage. For example, in Figure \ref{fig:rectangle}, 
it is easy to know that the matchings of edge $(b,c)$ are the subset of the matchings of edge $(a,c)$. When we generate partial matches of the two edges from data graph (Fig. \ref{fig:datagraph}), one common way is to materialize two edges using in-arcs of matching of $c$ separately. So, the data matching the two edges will be accessed twice. Another way is to initialize $(b,c)$ and $(a,c)$ via one access of data. Although it avoids repeat access, some unmatchings of $(b,c)$ will be hold on memory until they are pruned.

\textbf{Redundant Exploration}. Matchings of an edge in a pattern graph may be subset of another edge's matchings. As the size of the pattern graph increases, it implies more such relationships between edges. Current graph matching approaches ignore such relationships and check matchings of the edges again during their explorations. So, 
matchings of an edge are not reused for the exploration of another edge although the two edges have such relationship which indicates some exploration steps are not necessary. For example, in Fig. \ref{fig:rectangle}, the matchings of edge $(f,c)$ is subset of results of $(b,c)$ while some matchings of $(a,c)$ are matchings of $(b,c)$. The initial matchings of ($a$, $c$), ($b$, $c$), and ($f$, $c$) are 13 in-arcs of $c$'s matchings (vertex 3, 6, 8 in Fig. \ref{fig:datagraph}). Then, each of 13 in-arcs will be checked during the explorations of those edges. However, some matchings for $(f,c)$, which are not consistent with $(b,c)$, are checked again. Hence, there exists many redundant explorations. The redundant explorations also incur large number of partial matches resident in memory for further operations. When the data graph is big, the amount of redundant explorations and memory consumption grow rapidly. 

To address these problems, it motivates us to develop a directed subgraph matching approach which reduces a pattern graph so as to avoid redundant data access and exploration by identifying relevant constraints and reusing the matchings of relevant constraints. An important concept of our approach is the constraint inclusion relationship in which the matchings of a constraint (e.g. edge) are included by the matchings of another constraint. So a pattern graph can be reduced by temporarily removing some constraints included by other constraints.  
We take the constraint inclusion relationship to highlight a crucial feature of graph pattern matching problems. Following the idea, we take a 'pattern-reduction' approach towards building XMiner for efficient digraph matching. XMiner first reduces $\mathbb{G}$ by analysis of the inclusion relationship between constraints and then explores the reduced pattern graph. For example, in Figure \ref{fig:rectangle}, edge ($b$, $c$) and ($f$, $c$) belong to the constraint inclusion set of ($a$, $c$) because the matchings of edge ($b$, $c$) and ($f$, $c$) are included in ($a$, $c$). So, ($b$, $c$) and ($f$, $c$) can be removed temporarily. Partial matchings to the reduced pattern graph can be extended to the original pattern graph by materializing the removed constraints using previous results at suitable time. For example, ($b$, $c$) can be materialized using the results of ($a$, $c$) after some embeddings are removed due to degree. Similarly, the results of ($b$, $c$) can be extended to ($f$, $c$). Thus, XMiner avoids redundant data access and exploration of unnecessary edges.
Concretely,

\textbf{Pattern Graph Reduction}. By analyzing the relationship among constraints specified in a pattern graph, our approach finds a minimum set of constraints which cover other constraints. An algorithm is proposed to find a solution for the problem. By this way, our approach reduces the given pattern graph to a simpler form. Pattern reduction has the advantage that vertices/edges contained by a vertex/edge are materialized using smaller intermediate results instead of finding matchings from the data graph directly. Thus, it reduces access cost and memory usage.

\textbf{Efficient Execution}. XMiner generates an efficient execution plan to guide its exploration and materialization for a task. XMiner reuses the matchings of edges. So, it does not need to materialize each edge using data graph directly. This results in much lesser data access and computation compared to the state-of-the-art graph matching approaches. Moreover, XMiner postpones the initialization of edges until a suitable time. It reduces redundancy intermediate results, resulting in lesser memory consumption.
\begin{figure}
\centering
{\subfigure[Data Graph]
    {\includegraphics [scale=0.298] {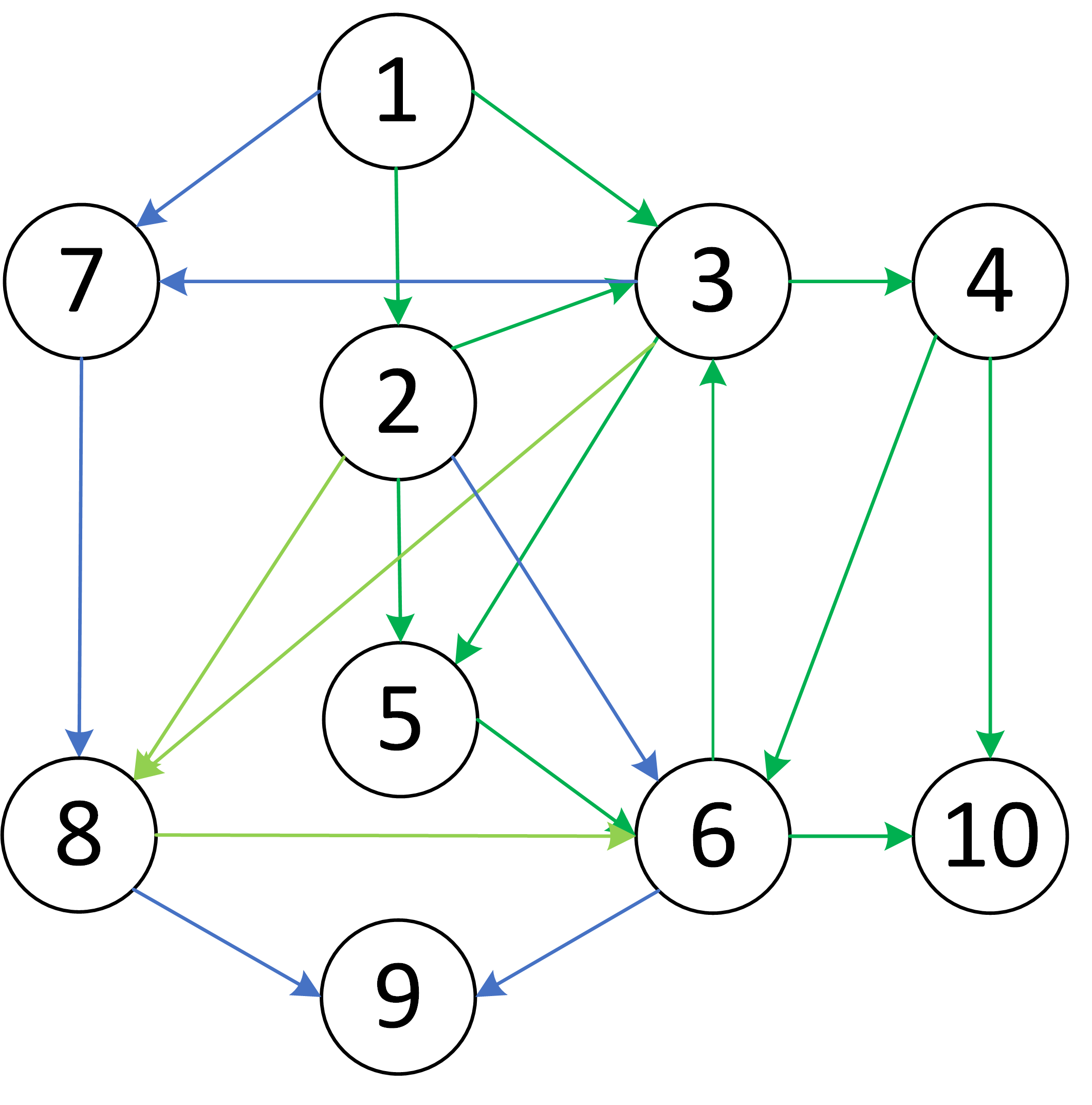}
    \label{fig:datagraph}}
}
\hfil
{\subfigure[Pattern Graph 7\_pa]
    {\includegraphics [scale=0.34] {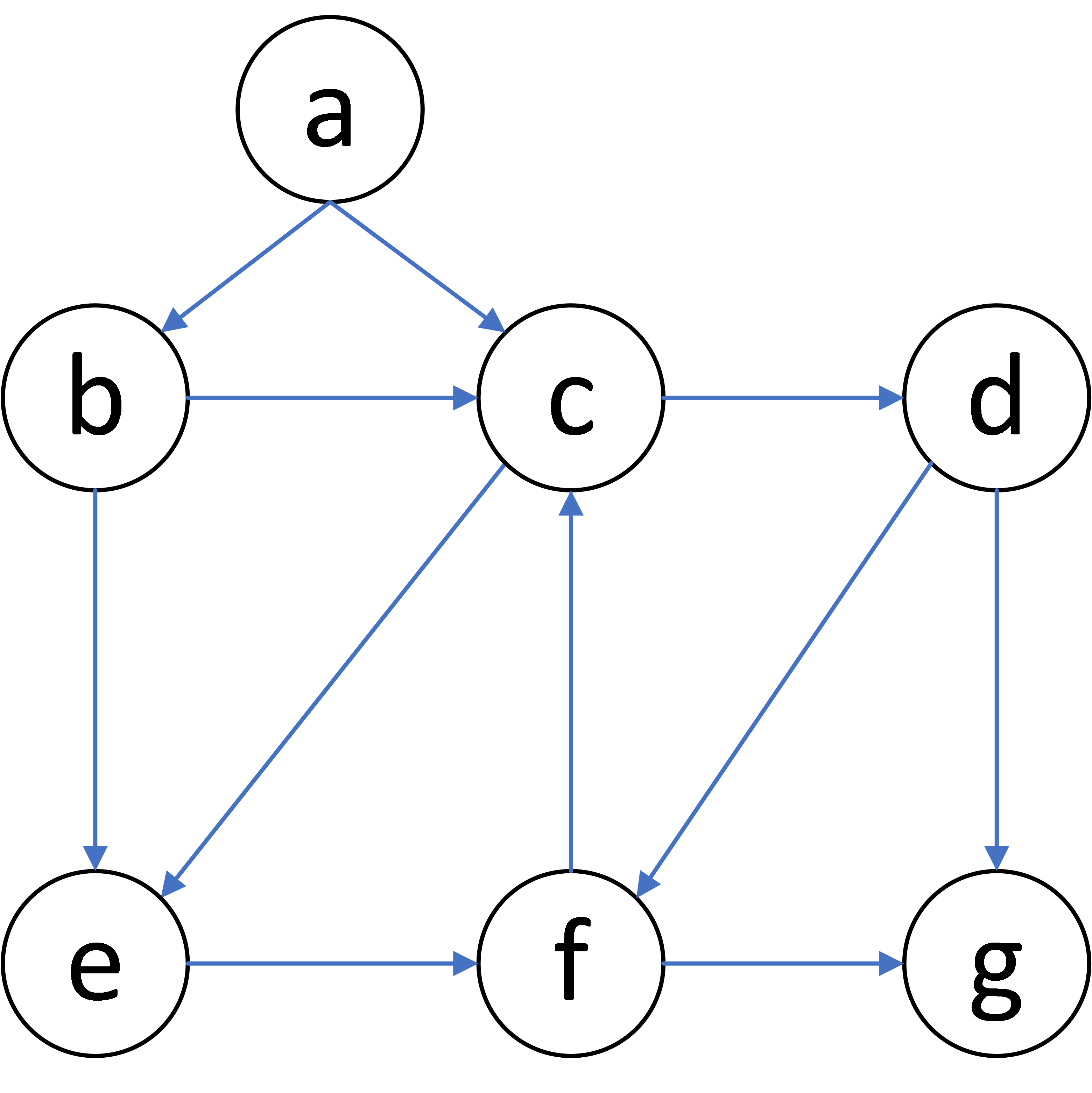}
    \label{fig:rectangle}}
}
\caption{Data Graph and Example Pattern Graph}
\label{fig:pattern}
\vspace{-15pt}
\end{figure}

\textbf{High Performance}. We demonstrate the efficacy of XMiner by evaluating it on 10 different size real-world graphs. Our experimental results shows that XMiner outperforms state-of-the-art undirected graph matching systems. Furthermore, XMiner can scale to large graphs and complex graph matching tasks which could not be efficiently handled by other graph matching systems.

\section{Preliminary}
\label{sec:preliminary}
A directed graph matching approach must ensure that the respective underlying pattern graph and its matches are isomorphic and are oriented the same. Before introducing our approach, we bring up some important concepts. 


\begin{definition}[Digraph]
\label{def:graph}
A directed graph $G = (V, E)$. ($u,v$) indicates a one-way arc with the direction from node $u$ to $v$. 
\end{definition}

A graph matching problem consists in the constraints (or patterns) which impose what nodes and arcs of data graph can be the matchings. 
The constraints are generally represented as a pattern graph.
\begin{definition}[Pattern Graph]
\label{def:patterngraph}
Let $G=(V,E)$ be the data digraph. Pattern graph $\mathbb{G} = (\mathbb{V}, \mathbb{E}, C)$ where $\mathbb{V}=V\bigcup \Upsilon$. $\Upsilon$ is a  variable set. The domain of each $e\in \mathbb{E}$ is $E$ of $G$. 
For $u\in \mathbb{V}$, its in-degree $deg^-(u)$=$\|\{v|(v,u)\in \mathbb{E}\}\|$ and out-degree $deg^+(u)$= $\|\{v|(u,v)\in \mathbb{E}\}\|$. $deg(u)=deg^+(u)+deg^-(u)$. The constraint set 
$C=\mathbb{V} \bigcup \mathbb{E}$ which specifies the allowable values of vertices and edges in $\mathbb{G}$.
\end{definition}


The constraints $C$ of pattern graph $\mathbb{G}$ includes two kinds of constraints: vertex constraint and edge constraint. The vertex constraint places a limitation on a vertex. Degree is a vertex constraint which limits any matching nodes in $G$ have degrees not smaller than the degree of corresponding vertex in $\mathbb{G}$.  
An edge constraint involves two adjacent vertices in $\mathbb{G}$. 
Each vertex of a pattern graph may appear in several edge constraints, which are adjacent constraints of the vertex. Each edge constraint $c_{j}$=($u,v$) specifies the allowable values of $G$ for the vertices appearing in $c_{j}$. Each matching of edge  ($u,v$) should be a link between each node pair in $G$ which matches $u$ to $v$ respectively.  
Fig. \ref{fig:pattern} shows an example data graph and an pattern graph. There are 11 edge constraints in Fig. \ref{fig:rectangle}. The domain of each edge constraint is all arcs of Fig. \ref{fig:datagraph}. However, 
the subgraphs with green lines in Fig. \ref{fig:datagraph} are matches of Fig. \ref{fig:rectangle}. 

Following the idea of relational algebra, we treat node set matching vertex $v$ as a one-column table (denoted as $A^v$) which keeps all the matchings of  vertex $v$. Similarly, we also treat arc set matching edge $e$ as 2-column relational table (denoted as $A^e$) where columns correspond to source and destination of edge $e$. $v_i$ may appear in several edges. The set of all allowable values of $v_i$ under edge constraint $e_j$ is the projection of $A^{e_j}$ over $v_i$, denoted by $\Pi_{v_i}(A^{e_j})$. As a pattern graph has other constraints, some matchings of constraint $c$ ($A^{c_j}$) may be removed due to the conflict with other constraints. Here, the final matchings are denoted as $\check{A}^{c}$. 
If edge constraint $e_{i}$, $e_{j}$ involving common vertices, the matchings of constraint $e_{i}$, $e_{j}$ is to join $A^{e_i}$ and $A^{e_j}$, denoted as $e_i\Join e_j$ for brevity. 

\begin{definition}[Path]
\label{def:path}
A path in $\mathbb{G}$ is a set of directed edges where the destination of an edge is and only is the source of another. Namely, $\forall k (0\leq k <m-2)$, $(v_k, v_{k+1}) \in \mathbb{E}$, edges ($v_0$, $v_1$), ($v_1$, $v_2$)..., ($v_{m-2}$, $v_{m-1}$) form a path between vertex $v_0$ and $v_{m-1}$. The path is also denoted as $<v_0, v_1,..., v_{m-1}>$, a sequence of vertices from one vertex to another using the edges. It can be also denoted as a ssequence of edges. 
If the destination of path $p_1$ is the source of path $p_2$, then we can concatenates $p_1$, $p_2$, denoted as $p_1\bowtie p_2$. The matches of path $P$: $<v_0, v_1,..., v_{m-1}>$ is $A^P$=
$\Join_{0\leq i<m-1} (v_i, v_{i+1})$.
\end{definition}

If the allowable values of a vertex or an edge is the subset of another vertex or edge (e.g. $A^{(b,c)}\subset A^{(a,c)}$), it is not necessary to materialize two vertices/edges independently because the intermediate results can be reused. So, in the following, we characterize relationship between result sets of constraints. 
\begin{definition}[Constraint Inclusion]
\label{def:const-rel}
Let $c_i$, $c_j$ be constraints of $\mathbb{G}$. If $A^{c_i}\subseteq$ $A^{c_j}$, then $c_j$ includes $c_i$, denoted as $c_i \sqsubseteq c_j$. 
If $c_i$, $c_j$ are vertex constraint (e.g. degree), the constraint relationship between them is vertex inclusion. If they are edges, the relationship is edge inclusion.
\end{definition}

Actually, if we ignore other constraints in a pattern graph, any two edge constraints are equal because their domains are all arcs of data graph. However, the matches of two edge constraints may be not same again when other constraints in the pattern graph are considered. So, in the following, we will explore constraint relationships in a given context.
\begin{definition}[Conditional Inclusion]
\label{def:incset}
For $v_i$, $v_j \in \mathbb{G}$, and constraint $c_k$, $c_l \in C$, if $\Pi_{v_{j}}A^{c_l} \subseteq \Pi_{v_i}A^{c_k}$, then $v_i$ conditionally includes $v_j$ respectively under constraint $c_k$, $c_l$, denoted by $(v_j)^{c_l} \sqsubseteq (v_i)^{c_k}$.

Similarly, under constraint $c_k$ edge $e_i$ conditionally includes edge $e_j$ under constraint $c_l$ if $\Pi_{e_{j}}A^{c_l} \subseteq \Pi_{e_i}A^{c_k}$. It is denoted as $(e_j)^{c_l} \sqsubseteq (e_i)^{c_k}$.

We can extend the context with more constraints. If there are two or more constraints, use semicolons to separate them. And, the constraint order indicates  the order in which a constraint is imposed on intermediate results. 
\end{definition}

In Fig. \ref{fig:rectangle}, since  edge constraint ($a, b$) and ($a, c$) share the same source, $(a,c)^{a} = (a,b)^{a}$ if we ignore vertex constraint $b$, $c$. However, when vertex constraint $b$ and $c$ are considered respectively,  $(a,c)^{a;c} \sqsubseteq (a,b)^{a;b}$ because  $c$'s degree is larger than $b$'s degree. So, some matchings of $A^{(a,c)}$ will be removed since they do not match $(a,b)$. It can be shortened to $(a,c) \sqsubseteq (a,b)$ when no misunderstanding occurs. Similarly, $c^{(a,c)} \sqsubseteq b^{(a,b)}$. A context also indicates an exploration order. Different exploration orders may produce different intermediate results, e.g. $(a,c)^{c} \neq (a,c)^{a}$ since $(a,c)^{a}$ means we reduce the embeddings of edge $(a,c)$ after the exploration of vertex  $a$.
\begin{theorem}
\label{lemma:vertexedgetrans}
 Let $e_1: (u_1,v_1)$, $e_2: (u_2,v_2)$ $\in\mathbb{G}$, if $u_1\sqsubseteq u_2$, then $e_1\sqsubseteq e_2$ denoted as $(e_1) ^{u_1}\sqsubseteq (e_2) ^{u_2}$.
 Similarly, if $v_1\sqsubseteq v_2$, then $e_1\sqsubseteq e_2$, denoted as $(e_1) ^{v_1}\sqsubseteq (e_2) ^{v_2}$.
\end{theorem}

According to Theorem \ref{lemma:vertexedgetrans}, there exist constraint inclusion relationships among any edges which share same sources or destinations. In Fig. \ref{fig:rectangle}, $(b,c)^{c}\sqsubseteq (a,c)^{c}$. However, $(c,e)\not\sqsubseteq (a,c)$ because their sources or destinations are different. 
We can extend Theorem \ref{lemma:vertexedgetrans} to the inclusion relationship of edge.

\begin{theorem}
\label{lemma:edgetransitive}
 Let $e_1: (u_1,v_1)$, $e_2: (v_1,w_1)$, $e_3: (u_2,v_2)$, $e_4: (v_2,w_2)$ $\in\mathbb{G}$, if $e_1\sqsubseteq e_3$, then $e_2\sqsubseteq e_4$. It is denoted as $(e_2) ^{e_1}\sqsubseteq (e_4) ^{e_3}$.
\end{theorem}
\begin{proof}
 Since $e_1\sqsubseteq e_3$, then $A^{e_1}\subseteq A^{e_3}$. $A^{e_2}$=$e_1\rtimes e_2$, $A^{e_4}$=$e_3\rtimes e_4$. Because the matchings of all edge patterns are arcs of data graph, so $A^{e_2}\subseteq A^{e_4}$.
\end{proof}

In Theorem \ref{lemma:edgetransitive}, $e_1$ and $e_2$, $e_3$ and $e_4$ are paths, respectively. We also say path $<e_1, e_2> \sqsubseteq <e_3, e_4>$. 
Let $P_1$, $P_2$ be paths, if $P_1 \sqsubseteq P_2$ and $P_2 \sqsubseteq P_1$, then $P_1 = P_2$. In Figure  \ref{fig:rectangle}, for path $P_1$:$<a, c, d>$ and $P_2$:$<a, b, e>$, $P_1$ $ \sqsubseteq$  $P_2$ and $P_2$ $ \sqsubseteq$  $P_1$ because ($a, b$)= ($a, c$) and $(c, d)^{(a, b)}$=$(b, e)^{(a, c)}$.

Constraint relationship indicates relationship between the matchings of edges or vertices. If a digraph matching system knows the constraint relationship, it can reuse intermediate results of other edge constraints instead of materializing each edge or vertex from the storage. Thus, it is necessary to identify all the constraints each edge or vertex includes.




\begin{definition}[Inclusion Set]
\label{def:incset}
Pattern graph $\mathbb{G}$ = ($\mathbb{V}$, $\mathbb{E}$, C). The vertex set included by $v (\in\mathbb{V})$ under $c_k$ is denoted as $I^{c_k}(v)$=\{$u^{c_l}|$ $u^{c_l} \sqsubseteq v^{c_k}$\}. Similarly, the  inclusion set of vertex $v$ under path $p_i$ is $I^{p_i}(v)$=\{$u^{p_j}|$ $u^{p_j} \sqsubseteq v^{p_i}$\}.
The edge inclusion set of edge $e(\in \mathbb{E})$ under $c_k$ is defined as $I^{c_k}(e)$=\{$x^{c_l}|$ $x\in \mathbb{E}$, $x^{c_l} \sqsubseteq e^{c_k}$\}.  
All the edges included by $e$ is defined as $\mathcal{I}(e)$=$\bigcup I^{c}(e)$. 
\end{definition}

For instance, in Figure \ref{fig:rectangle}, under edge constraint ($a,c$), $I^{(a,c)}(a)$= \{$b^{(b,c)}$, $f^{(f,c)}$\}. $I^c(a,c)$= \{$(b,c)^c$, $(f,c)^c$\}.

Here we mainly concern about edge constraint because it implies restrictions on two endpoints. In order to reuse intermediate results and reduce data access, we should find an edge set (i.e. edge Constraint cover) so that all edges in the set contain all the edges of $\mathbb{G}$.
\begin{definition}[Constraint Cover]
\label{def:incset}
Pattern graph $\mathbb{G}$ = ($\mathbb{V}$, $\mathbb{E}$, C). Let $\mathbb{S}\subseteq \mathbb{E}$. $\mathbb{S}$ is an edge constraint cover of $\mathbb{G}$ iff $\forall e\in \mathbb{E}$, $\exists s\in \mathbb{S}$, $e \in \mathcal{I}(s)$.
\end{definition}

Our objective is to find the minimum constraint cover $\mathbb{S}$ such that the amount of data access is minimal. The  constraint covering problem is defined as follows:  $ \underset{\mathbb{S}}{\arg\min}(\sum_{e\in \mathbb{S}} ||A_{e}||)$ (s.t. $\cup_{e\in \mathbb{S}}\mathcal{I}(e)=\mathbb{E}$). Finding minimum constraint cover is a set cover problem, which  is a NP-complete problem. 
In Section \ref{sec:reduction},  an algorithm for the problem is presented.
\section{Overview of XMiner}
\label{sec:overview}
The core idea of XMiner is pattern reduction. With the approach, XMiner can handle complex pattern graphs such as those that have more vertices and edges. In the process of matching the pattern graph, we will face the problems of how to reduce the pattern graph and restore results.

To solve the problem of pattern reduction, we first need to identify the constraint relationship. 
XMiner discovers constraint relationships of edges and vertices according to the concepts in Section \ref{sec:preliminary}. After identifying constraint inclusion sets of a pattern graph $\mathbb{G}$, XMiner reduces $\mathbb{G}$ into a subgraph $\mathbb{G}'$. Note that each edge in $\mathbb{G}'$ corresponds to a set of edges of $\mathbb{G}$ included by it. 
This will be elaborated in Section \ref{sec:reduction}.

XMiner generates an exploration order for $\mathbb{G}'$. Then, the exploration process identifies start nodes of data graph matching the first vertex in pattern graph $\mathbb{G}'$ and then performs parallel search from start nodes. When the search continues, we obtain candidate matchings of the edges or vertices from data graph. If the constraints includes other constraints, then those constraints contained by the constraint can reuse the matchings of the constraint at suitable time. 
This will be elaborated in Sec. \ref{sec:mining}.
\vspace{-10pt}
\section{Pattern Graph Reduction}
\label{sec:reduction}
The idea of pattern reduction is to analyze constraint inclusion relationships and then temporarily remove the edges included by other edges. 

\subsection{Finding Constraint Inclusion Relationship}
\label{sec:constinc}
To reduce the pattern graph, we first compute the constraint inclusion sets for all vertices and edges. Vertices depends on other through corresponding edges. So, we can identify their relationship according to edges. Actually, discovering constraint inclusion sets is to compute the transitive closure of each vertex and edge in the pattern graph. This is accomplished by traversing $\mathbb{G}$ sequentially and constructing the inclusion set of each vertex or edge according to the above theorems in Section \ref{sec:preliminary}. 

XMiner first identifies constraint inclusion relationship of each vertex's neighbors. 
For each vertex $v$, it chooses $v$'s in-neighbor with smallest degree (say $u$) because its matchings include the matchings of other in-neighbors of $v$. Then XMiner adds $v$'s other neighbors to $u$'s inclusion set. At the same time, the in-edges of $v$ are also inserted into the inclusion set of edge ($u,v$).  Similar operations are also applied to the out-neighbors and out-edges of $v$. After identifying neighbors' inclusion sets, XMiner will explore inclusion relationships through paths according to Theorem \ref{lemma:vertexedgetrans} and \ref{lemma:edgetransitive}. 
We grow the constraint inclusion set along traversal path. The step is repeated until no constraint inclusion set is updated. The final result is the closure of inclusion sets. 

Consider the example in Fig.  \ref{fig:rectangle}, we first identify inclusion sets of each vertex and its edges. Vertex $c$ has five neighbors: $a$, $b$, $e$, $f$, and $d$. There are two out-edges from $c$, to $e$ and $d$, respectively, three in-edges from $a$, $b$, and $f$ to $c$. Thus, we first have two edge inclusion sets: $(c, e)^c$ = $(c, d)^c$ and $(a, c)^c$ = $(b, c)^c$=$(f, c)^c$. Further, since $deg(a)<$ $deg(b)<$ $deg(f)$, edge $(f, c)^{c;f}$ $\sqsubseteq$ $(b, c)^{c;b}$ $\sqsubseteq$ $(a, c)^{c;a}$. 
The inclusion set of $a$ includes vertex $b$, and $f$. The same procedure is now applied for other vertices, so that we can identify their inclusion sets and neighbor edges' inclusion sets. Then we traverse the pattern graph to update inclusion sets. For example, we can know $(c,d)^{(b,c)}$ =$(e,f)^{(b,e)}$=$(c,e)^{(b,c)}$ according to Theorem \ref{lemma:vertexedgetrans} and \ref{lemma:edgetransitive}. In other words, $<b,c,d>$ =$<b,e,f>$=$<b,c,e>$. Each member $c_j$ in $c_i$'s inclusion set indicates that all matchings for $c_j$ are included in $c_i$'s intermediate results under the given context. Fig. \ref{fig:rectangle} and \ref{fig:7pplan} shows some inclusion relationships. Edges or vertices with same number (e.g.\Mt{black}{\textcolor{white}1} for vertices or \Mt{white}{1} and \Mt{black}{\textcolor{white}1} for edges) indicate there exist inclusion relationships among them. 
\subsection{Reducing Pattern Graph}
As mentioned in Section \ref{sec:preliminary}, one objective of pattern graph reduction is to find the minimum constraint cover $\mathbb{S}$ such that the amount of data access is minimal. However, it is hard to estimate the total amount of data access accurately since it is actually a selectivity problem. Exhaustive enumeration of all exploration orders of a pattern graph is much too slow. Since the domains of edges and vertices of pattern graph are data graph, the matchings of the edge are fewer if one endpoint of an edge has more restrictions. So, we design a heuristic algorithm (Algorithm \ref{alg:decompattern}) in which 
we reduce a pattern graph by searching the graph from a start vertex with maximal degree and removing the constraints included in other constraints step by step.
Formally, a pattern graph can be reduced in the following manner: i) choose the most selective vertex as a start, then select its edges with the highest coverage and insert the edge into a list. ii) update the constraint inclusion set of another endpoint of the edge. iii) if all edges of the vertex are checked, Its neighbors with maximal degree will be processed. iv) if there exist unvisited edges in the pattern graph, the above process will repeat again.

\begin{algorithm}[ht]
    \caption{Reduce Pattern Graph}
    \label{alg:decompattern}
    \KwIn{Pattern graph $\mathbb{G} = (\mathbb{V}, \mathbb{E}, C)$; The inclusion set $\mathcal{I}$ of each edge and vertex;}
    \KwOut{matching plan $p$;}

    $p.start=\mathbb{G}$.maxDegree(); \label{alg:code:choosestart}

    $p.split=null$;

    searchSplit($p.start$, $p.split$);

    \While{$\exists e\in \mathbb{G}$, $e.visited\neq true$}
    {
        $u=(deg(e.src)\leq deg(e.dest))?e.dest:e.src$;

        searchSplit($u$, $p.split$);
    }
    \BlankLine
    Procedure \textbf{searchSplit}(in $u$, out $s$)

    $v=u$;

    \Repeat {$v$ IS $NULL$}
    {

        \While { $\exists e$  ($e.src==v~ OR~ e.dest==v$) AND ($e.visited \neq true$)}
        {

           $mec$=findMaxECoverage($v$); \label{alg:code:findmaxecoverage}


           $mec.visited$=$true$;

           $s$.enque($mec$); \label{alg:code:enque}

           $w=(mec.src\neq v)?mec.src:mec.dest$;
           updateVInclusion($w, mec$);


        }
        $v$=Neighbor\_MaxDegree(w); \label{alg:code:searchmaxngb} 
    }

    \BlankLine
    Procedure \textbf{updateVInclusion} (in $w$, in $ec$)

        \For {each path P s.t. $( I^{P}(w)\in \mathcal{I}(w))$ }
        {

             remove $ec$ from $P$;

             \For {each path $P'$ s.t. $(y^{P'}\in I^{P}(w))$ }
             {
                  \If {$P'=(v,x)||(x,v)$ and unvisited}
                  {
                       set $P'$ as visited;


                       updateVInclusion($v||x, P'$);
                  }
             }
        }
\end{algorithm}
The algorithm first chooses the vertex having a maximal degree as the start (Line \ref{alg:code:choosestart}). That is, we favor the vertex which has the smallest number of candidate nodes among the vertices in the pattern graph. If more than one vertex has the same maximum degree, we randomly pick one. Consider the pattern graph (Figure \ref{fig:rectangle}) and the data graph (Figure \ref{fig:datagraph}). Due to the degrees of vertices, the possible cardinalities for vertex $a$, $c$ are respectively 7 and 3. If we perform the subgraph search using vertex $a$ first, a lot of partial solutions will be enumerated. However, many of them will be finally discarded due to unmatching with $c$, since the data nodes can not satisfy the constraints (such as degrees) specified in the pattern graph. On the other hand, if we perform the subgraph isomorphism search using $c$, we can prune more impossible matchings.

The algorithm then reduces the pattern graph by finding a constraint cover from the start vertices. The reduced pattern graph is stored in a queue $split$. The algorithm calls procedure \textit{searchSplit} which finds the branches of reduced graph from vertex $v$. \textit{searchSplit} then checks whether there exist unvisited adjacent edges of $v$. If there exist unvisited edges, the algorithm chooses the edge with a high coverage (Line \ref{alg:code:findmaxecoverage}). Then, it adds the selected edges to the queue $s$ which stores the edges and is null initially (Line \ref{alg:code:enque}). The edges in the queue are set as visited and will not be checked again. So, the edge is temporarily removed from the pattern graph. It will influence the constraint inclusion sets containing the edge. \textit{searchSplit} will update the constraint inclusion sets by invoking \textit{updateVInclusion}. For each path in the constraint inclusion sets, \textit{updateVInclusion} will remove the edge from the path. If the path is an edge, the algorithm then labels the edge as visited. If there exist unvisited edges, the algorithm will choose one endpoint with maximum degree of that edge to start a new search (Line \ref{alg:code:searchmaxngb}).
The same procedure repeats on the resultant graphs until all edges are visited. We can continue reduction of the pattern graph in this way. Therefore, when the algorithm has no more unvisited edges of the pattern graph to test, the process will terminate.

Along with this reduction, removed edge are recorded. The reduced pattern graph obtained includes vertices and partial edges linking these vertices, and the degrees on the vertices are directly taken from the original pattern graph. The edges removed are covered by the remained edges. Therefore, any constraints not covered will be reserved in the reduced pattern graph, while at the same time, the pattern graph is reduced to a smaller size,
which makes finding a larger complex pattern graph more easily. During the pattern graph reduction process, we also obtain a exploration order for the constraints.

\begin{figure}
\centering
{\subfigure[7\_pa]
    {\includegraphics [scale=0.3] {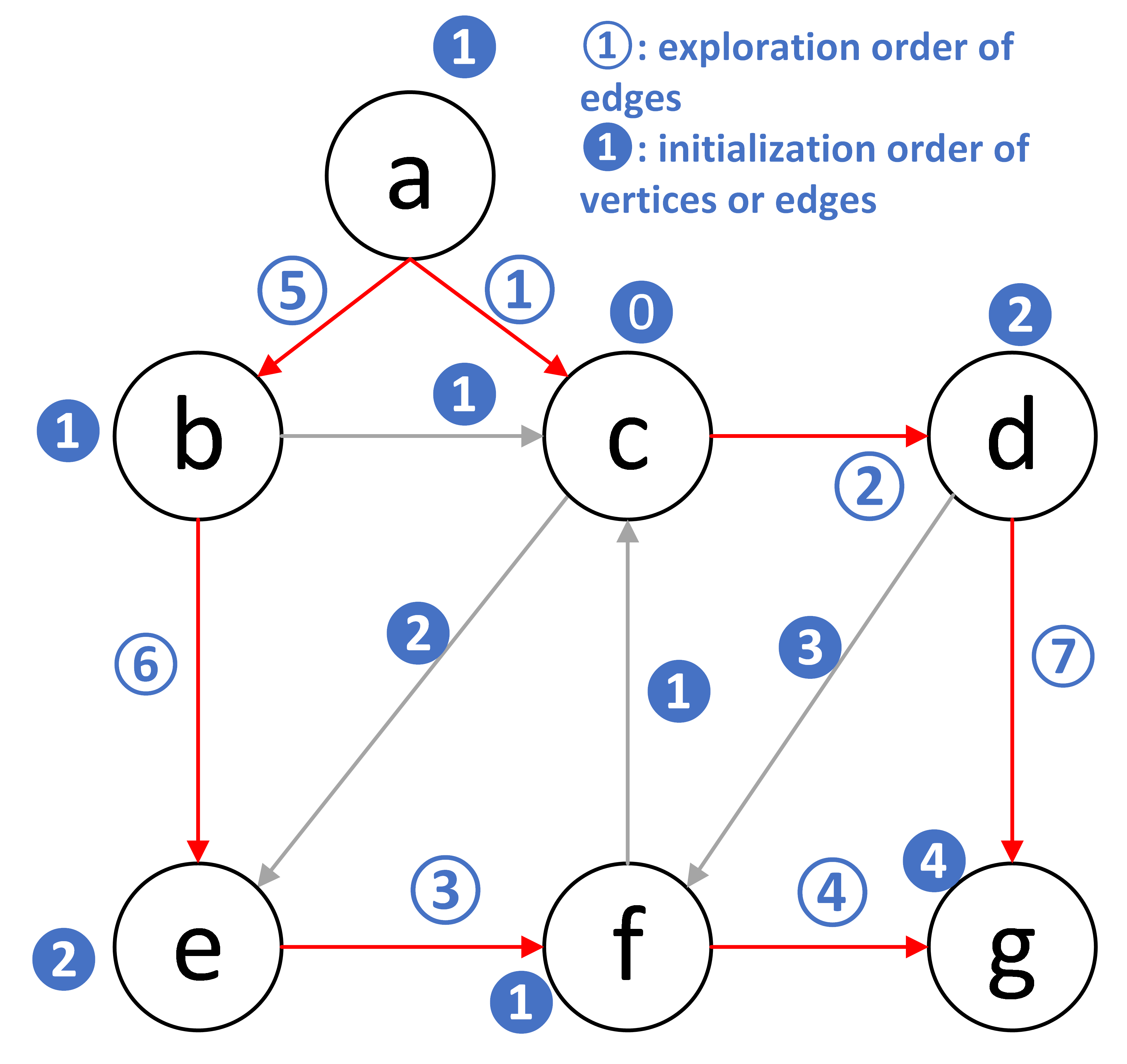}
    \label{fig:rectplan}}
}
\hfil
{\subfigure[7\_pc]
    {\includegraphics [scale=0.27] {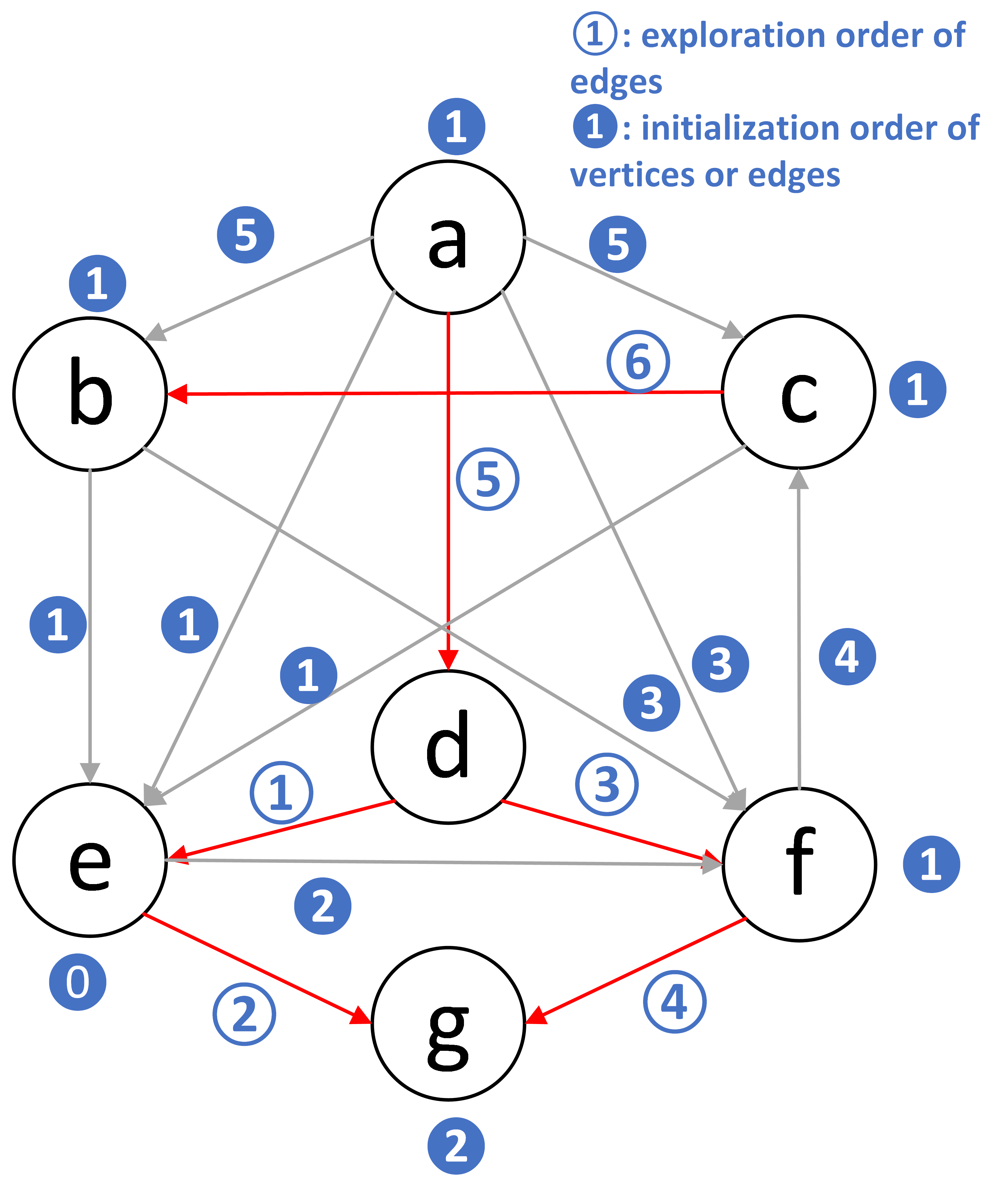}
    \label{fig:7pplan}}
}
\caption{ Reduced Pattern Graphs and Their matching Plans}
\label{fig:plan}
\vspace{-12pt}
\end{figure}
For example, in Figure \ref{fig:rectangle}, we select $c$ to start the reduction because it has the highest degree.  Its neighbor edges are included in two edges ($a$, $c$) and ($c$, $d$), respectively. So, edges ($a$, $c$) and ($c$, $d$) are inserted into the queue. The same procedure is now applied to $f$, $b$, and $d$ so that we add the edges  ($e$, $f$), ($f$, $g$), ($a$, $b$),  ($b$, $e$),and ($d$, $g$) to the queue.
XMiner reduces $\mathbb{G}$ (Figure \ref{fig:rectangle}) into $\mathbb{G}'$ (red lines in Figure \ref{fig:rectplan}). There is a branch where the edges are marked with red solid lines and the edge constraints included by others with gray dashed lines. Figure \ref{fig:7pplan} is the reduced pattern graph (shown in red lines) of pattern graph 7\_pc used in the experiments (Figure \ref{fig:7p}). In the pattern graph, there are more constraint inclusion relationships.
 So, six edges are left and 9 edges are removed temporarily. Each edge $e_i$ (red line) in $\mathbb{G}'$ corresponds to a set of edges (gray lines) in $\mathbb{G}$ which are included by $c_i$ if they have same number label (e.g. \Mt{white}{1} and \Mt{black}{\textcolor{white}1}). 

The algorithm finds a constraint cover by choosing the adjacent edges of a given vertex with the high coverage each step. Thus, the algorithm runs in $O(||\mathbb{E}||)$ time. Along with this reduction, the pattern graph is simplified. Note that sometimes we cannot reduce a pattern graph 
because there do not exist inclusion relationship among its constraints. Such pattern graphs are called irreducible. 
Note that a reduction never remove constraints from the pattern graph forever, but decreases the cost to process some constraints included in other constraints. So, the reduction does not lead to incorrect matching results.
\section{Pattern Matching}
\label{sec:mining}
After the pattern graph is reduced, consequently, the problem of graph matching would be to  find the matchings of the reduced pattern graph. XMiner first seeks for matchings for the reduced pattern graph, and extends the matchings to the edges removed temporarily according to the plan. 

\subsection{Plan Generation and Execution}
\label{sec:exec}
We need to develop a plan for the reduced pattern graph to direct our exploration in the data graph. The execution plan includes two parts: one part indicates an materialization order of vertices or edges, the other part defines the exploration order of edges. 
We choose the same start vertex as the reduction process to begin the generation of the execution plan. The start vertex will be first materialized. Then its adjacent edges with maximal inclusion set will be explored separately.  The edges included by them and the neighbor vertices of the start vertex will be correspondingly materialized. The status of the start vertex and the edges are set to be 'visited'. 
For the unvisited edges, we choose the vertex which is adjacent with the 'visited' vertices and has the maximal degree to start the above process. When all edges in the reduced pattern graph are visited, the plan for the pattern graph is generated. 
This plan guides the exploration of data graph to ensure generated matches are unique. 

For our example pattern graph (Figure \ref{fig:rectangle}), we choose $c$ as the start vertex and then explore edge $(a,c)$ and $(c,d)$ respectively. The edges included by them will be materialized using their intermediate results. Then we choose vertex $f$ to start the next round. The process will  continue until all edges are visited. Its plan is shown in Figure \ref{fig:rectplan} where black circled number (e.g.\Mt{black}{\textcolor{white}1}) indicates the materialization order of vertices or edges, and white circled number (e.g.\Mt{white}{2}) represents the exploration order of edges. In the plan, The ($k$-1)'th step is executed before $k$-th step. The reduced graphs are not larger than original pattern graphs and thus the plans are flatter than the search trees generated using DFS.
\begin{algorithm}[htb]
    \caption{Plan Execution}
    \label{alg:planexec}
    \KwIn{$p$ is the plan;}
    init $A^{p.start}$;

    \For {each $w \in A^{p.start}$ in Parallel}
    { \label{alg:code:paraexplbegin}
        $w.visited =true$;

        \textbf{Explore}($p.split$, $0$, $p.start$, $w$) ;
    } \label{alg:code:paraexplend}

    \BlankLine

    Procedure \textbf{Explore}(in $sp$, in $i$, in $vt$, in $nd$);

        $e$=$sp[i]$;

        \If{e.src==vt}
        {\label{alg:code:initbegin}
            $A^e= nd.out\_arcs$;

            $A^{e.dest}=(||A^{e.dest}||==0)? \prod_{e.dest} A^e$: $A^{e.dest}\cap \prod_{e.dest} A^e$;
        }
        \Else{
            $A^e= nd.in\_arcs$;

            $A^{e.src}=(||A^{e.src}||==0)? \prod_{e.src} A^e$: $A^{e.src}\cap \prod_{e.src} A^e$;
        } \label{alg:code:initend}

        \For{each $I^{P_1}(u||v)\in \mathcal{I}(u||v)$ and $e\in P_1$}
        { \label{alg:code:propbegin}

            \For {$y^{P_2}\in I^{P_1}(u||v)$ }
            {
                \If {$P_2=(u||v,x)||(x,u||v)$}
                {
                    \If {$A^x$ is NULL}
                    {
                        $A^{P_2}=A^e$;
                        $A^x=A^{u||v}$;
                    }
                    \Else{
                        $A^{P_2}\ltimes A^e$; $A^{P_2}\rtimes A^e$;

                        $A^x=A^x\cap A^{u||v}$;
                    }

                    update all inclusion set by removing $P_2$ from constraints;
                }
            }
            remove $e$ from $P_1$;
        }  \label{alg:code:propend}
        $nexte=sp[i+1]$; \label{alg:code:nxtexplbegin}

        \If { ($||A^{nexte.src}||\neq null$) OR ($||A^{nexte.src}||\leq ||A^{nexte.dest}||$) }
        {
            $maplist=A^{nexte.src}$;
            $u=nexte.src$;

        }
        \Else{
            $maplist=A^{nexte.dest}$;
            $u=nexte.dest$;

        }

        \For{each $w\in maplist$ AND $w.visited \neq true$}
        {
            $w.visited =true$;

            \textbf{Explore}($sp, i+1, u, w$);
        }\label{alg:code:nxtexplend}
\end{algorithm}
\vspace{-3pt}

After obtaining the execution plan for the pattern graph, XMiner follows the plan to start the task (Algorithm \ref{alg:planexec}). The first step is to materialize the start vertex using its candidate set derived from the data graph. The mappings of the start vertex should satisfy the degree constraint. 
XMiner then distributes the candidate nodes across threads so that exploration tasks are executed in parallel (Line \ref{alg:code:paraexplbegin}-\ref{alg:code:paraexplend}). Each thread starts its explorations of the data graph from each of the nodes assigned to it one by one. By this way, exploration tasks are mutually independent and convenient for scheduling. So, XMiner does not need synchronization when tasks are executed in parallel by multiple threads. It leads to better multithreading scalability. To avoid stragglers and maximize parallelism, XMiner schedules exploration tasks using work-stealing algorithm \cite{Burton81,worksteal1999} because it can reduce the workload imbalance incurred by the varied search space of exploration tasks.

XMiner subsequently selects first pattern edge $e_i$ given by the matching plan. For each start node $u$, XMiner retrieves matchings of $e_i$ in the data graph whose endpoints are matchings of $u$. After that, XMiner materializes another endpoint of $e_i$ using the matchings of $e_i$ (Line \ref{alg:code:initbegin}-\ref{alg:code:initend}). If edge $e_i$ includes other edges, XMiner will initialize those edges and their endpoints using the matchings of $e_i$ (Line \ref{alg:code:propbegin}-\ref{alg:code:propend}). Then XMiner will get the next edge in the plan and evaluate the matching size of its two endpoints. It will choose the endpoint with fewer matchings and explore its matchings one by one (Line \ref{alg:code:nxtexplbegin}-\ref{alg:code:nxtexplend}). The thread will continue the above process until all the matchings of the start vertex assigned to it are explored. 

\begin{figure}
\centering
\includegraphics [scale=0.5] {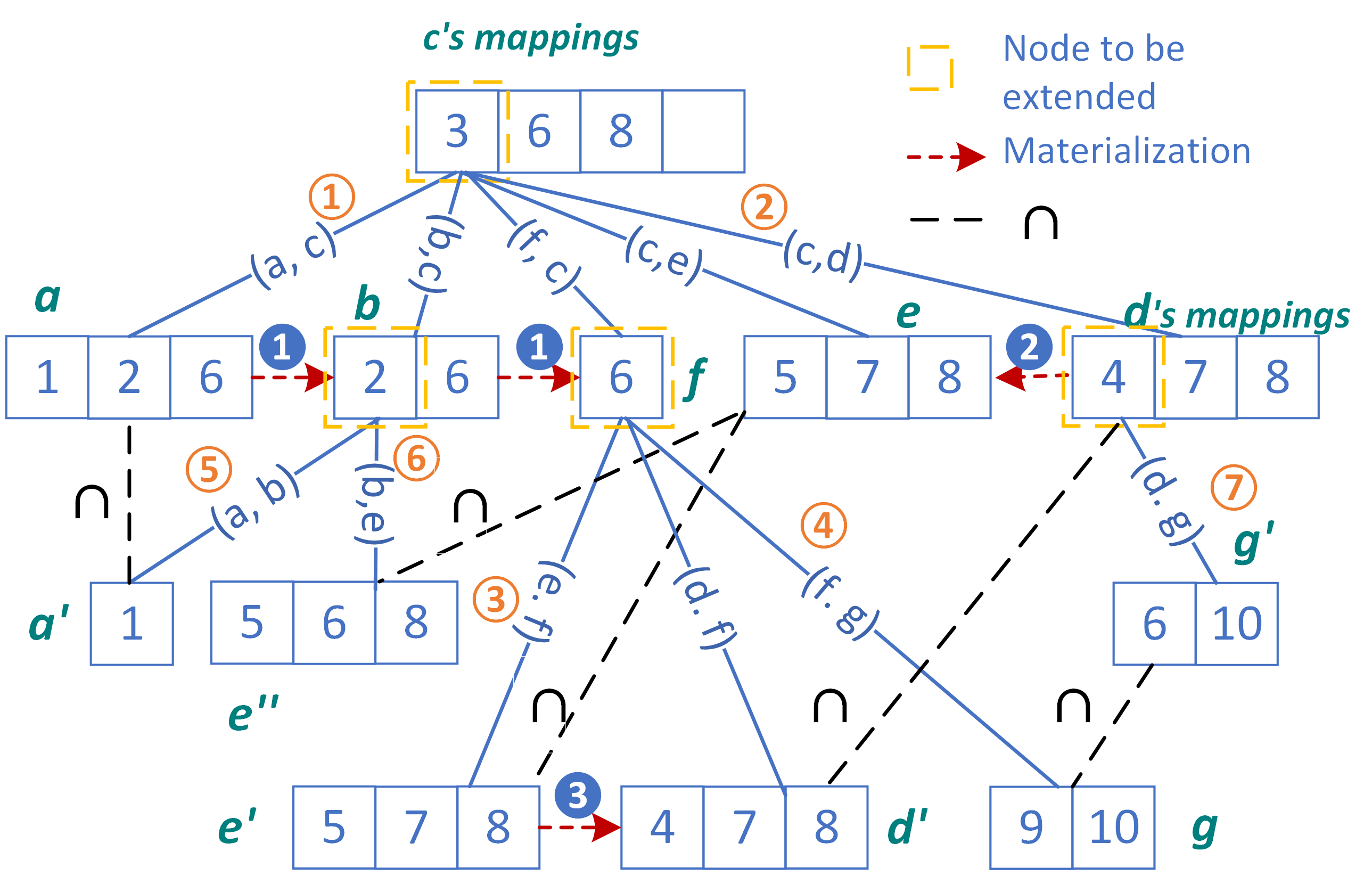}
\caption{The Example Plan Execution}
\label{fig:example}
\vspace{-10pt}
\end{figure}

Considering the running example shown in Figure \ref{fig:rectplan}, XMiner first starts with vertex $c$. XMiner gets nodes from the data graph (Figure \ref{fig:datagraph}) which satisfy condition $c.deg \leq v.deg (v\in G)$ (Step \Mt{black}{\textcolor{white}0}). So, node 3, 6, 8 are initial matchings of $c$ (Figure \ref{fig:example}). Each thread begins the explorations of data graph from the start nodes assigned to it. Here, for simplicity, assume three start nodes are assigned to a thread. The thread will begin its explorations from the start nodes one by one. It first searches the data graph from node 3. In Step \Mt{white}{1}, XMiner finds the matchings of edge ($a, c$) linked to 3 in the data graph. There are 3 arcs matching edge ($a, c$). So, we can get the partial matches of vertex $a$ (Step \Mt{black}{\textcolor{white}1}). Since edge ($b,~c$) and ($f,~c$) are in the inclusion set of edge ($a, c$), so, the matchings of edge ($a, c$) contain the two edges' matchings. As vertex $a$ includes vertex $b$ and $f$ under the corresponding constraints, the matchings of vertex $a$ can be assigned to vertex $b$ and $f$ . When assigning $b$ and $f$, XMiner deletes any matchings that are inconsistent with the constraints for $b$, $f$, respectively (Step \Mt{black}{\textcolor{white}1}). In the next step, XMiner explores edge ($c, d$) and materializes edge ($c, d$) and vertex $d$. After that, XMiner initializes $e$ according to the matchings of $d$. 
In Step \Mt{white}{3} and \Mt{white}{4}, XMiner will search the matching of vertex $f$, i.e. node 6, and then the matchings of $e^{\prime}$ and $g$ are obtained. In Step \Mt{white}{5}-\Mt{white}{7}), the thread will find the matchings  $a^{\prime}$, $e^{\prime\prime}$, and  $g^{\prime}$ by exploring the corresponding edges (For simplicity in presentations, here we temporarily let $a$ denote $A^a$). Following edge ($a, b$), XMiner update the partial matches of $a$ by intersecting two result sets of $a$. Since $a=a\cap a^{\prime}$ is not null and there exist matchings, the process will continue. Then, XMiner computes $e=e\cap e^{\prime\prime}$ (\Mt{white}{6}), and $g=g\cap g^{\prime}$ (\Mt{white}{7}). So, the final matchings of all the vertices are obtained under the given values (shown in yellow box). If XMiner outputs the matchings, the matchings will be combined for output (Section \ref{sec:generation}). Then, XMiner will search other candidate values (e.g. let $b$=6) until all the values are checked. 


\subsection{Intermediate Result Reuse}
During the exploration of data graph, partial matches generated in previous steps will be extended in subsequent steps. So, the amounts of intermediate matches grows rapidly and will be huge. Graph matching systems generally save the intermediate matches for further check. It places a heavy burden on system performance. The idea of constraint inclusion provides an efficient method of result reuse to limit redundant data access, memory usage, and reduce unnecessary computations. Inclusion constraint is first applied as a filter step to delete inconsistent matchings before the beginning of the next extension process.
When XMiner generates a partial matches of an edge, the result is used to materialize other  edges included in its inclusion set at subsequent time. 

Inclusion constraints also  spread the implications of a constraint on an edge or vertex onto other edges or vertices, respectively. For example, edge $e_i$ may be indirectly included by edge $e_k$ (e.g. edge $e_i \sqsubseteq e_j\sqsubseteq e_k$). After $e_k$ is initialized, one way is to materialize $e_i$ and $e_j$ using the matchings of $e_k$. This way incurs redundant computation. So, XMiner initializes $e_i$ using the matchings of $e_j$ instead of $e_k$ after $e_j$ is materialized using $e_k$. 
The method ensures that the values inconsistent with the domain of edges or vertices will not be processed twice. 
Compared to other graph matching systems which initialize once per vertex/edge, our method enables the initialization of vertices or edges in a inclusion set to be executed just once. It avoids much data access, rewrites, and large memory consumption. 
This way can eliminate duplicate storage and redundant computations for intermediate results and improve the performance.

For example, in Figure \ref{fig:rectplan}, 
$f^{(f, c)}$ $\sqsubseteq$ $b^{(b, c)}$ $\sqsubseteq$ $a^{(a, c)}$. After materializing $a$, one general way is to copy its matchings to the buffers for $b$ and $f$. 
However, XMiner does not materialize the vertices immediately, but does it after pruning some inconsistent results. Here, after removing the values of $a$ which are inconsistent with degree constraint of $b$, XMiner assigns the rest of $a$'s matchings to $b$ because constraint $(a, c)$ contains $(b, c)$. 
Similarly, we assign values of $b$ to $f$ after removing inconsistent values with $f$ from $b$. The reason to materialize $f$ using $b$ is that the matchings of $f$ is closer to $b$'s matchings. Although $d$ does not include $e$, edge $(c, d)^{c}$ $\sqsubseteq$ $(c, e)^{c}$. So we materialize both $d$ and $e$ using the matchings of $(c, d)$.
 \vspace{-12pt}
\subsection{Generating Results}
\label{sec:generation}
The matching results in the corresponding data graph must be completely consistent with the given pattern graph. 
The final results should be without repetition. Moreover, 
the mappings of some edges may not exist in the data graph. So, XMiner needs to check each arc and terminates the further exploration if the arc is not a candidate. It is important because the inspection process will guarantee the correctness and unique of the matching results.  

When executing the matching plan (Section \ref{sec:exec}), XMiner will extend the matchings of the chosen vertices (e.g. vertex $c$, $f$, $b$, and $d$). For each node to be extended, XMiner obtains the candidate sets of the rest pattern edges and vertices from the data graph. Next, the candidate sets of these vertices need to be combined to obtain the correct matchings of the pattern graph. 
When generating matchings, XMiner selects them in candidate sets and forms results. Suppose $I_j~(j\leq |\mathbb{V}|)$ is a candidate set of vertex $v_j$. We get a value from each $I_j$ and combine them. By this way, XMiner repeatedly generates all results for the candidate mappings. 
Since XMiner stores mapping of each edge of the pattern graph, so it need not check each arc between two nodes. The current combination results must be valid matchings.

Some graph matching problems need not output the enumerations of all embeddings, but the total number of embeddings. In this case, XMiner just multiplies the sizes of all $I_j$ for the combination of current active nodes and sums all of them. This way can speedup the matching process because it need not hold huge amount of intermediate results in memory for combination.

Consider the example in Figure \ref{fig:example}, 
the current matchings for $c$, $b$, $f$, and $d$ are shown in yellow box. For the current matchings, we will obtain the value set of $a$, $f$, and $g$. So, in this context, XMiner combines the values of $a$, $f$, and $g$ with the current matching nodes of $c$, $b$, $f$, and $d$. So, we obtain two matchings from the example pattern graph (subgraphs with green lines in Figure \ref{fig:datagraph}). XMiner will search data graphs for matchings of $a$, $f$, and $g$ with next combinations of the values of $c$, $b$, $f$, and $d$. The process will terminate until all combinations are evaluated. 
\vspace{-5pt}
\subsection{Implementation}
XMiner is implemented using C++ and Intel$^\circledR$ oneAPI Toolkits, compiled using gcc 7.5.0. 
We run XMiner in parallel by partitioning the matchings of start vertex of the given pattern graph over multiple threads. Concurrent threads then operate on those exploration tasks, each starting at a set of different nodes in the data graph. 
Each thread follows the plan to start the exploration in data graph. During the exploration process, each thread saves possible mappings for each edge and vertex of pattern graph and removes those invalid intermediate results. Since many threads search the input graph in parallel, it is highly possible that a thread may explore those which have been searched by other threads. In order to avoid redundant search, each thread maintains the lightweight boundary information regarding its exploration tasks. When the exploration reaches the boundary, the thread will not continue the search along the direction and turn to other directions for search if there exist un-visited areas.

XMiner allows users to directly express pattern graphs by giving all edges. Then XMiner analyzes the pattern graph, reduces the pattern graph and generates the efficient matching plan for the execution. By this way, users can describe pattern graph easily and do not care about how to efficiently handle it. It is useful for users unfamiliar with programming.

\section{Evaluation}
\label{sec:eval}
\subsection{Experimental Settings}

\textbf{\textit{Competitors.}} Although there emerge many graph matching algorithms and systems, they are designed for un-directed graphs or specific hardware \cite{PBE,GSI,PBE2,GAMMA,G2Miner} and no graph matching systems for directed graph on a single machine are available. So, we choose un-directed graph matching systems as competitors even though matching directed graph is more complex than matching un-directed graph as we analyze in Section \ref{sec:intro}. Undirected graph matching systems Peregrine \cite{Peregrine} and GraphPi \cite{GraphPi} are much faster than previous systems such as Arabesque \cite{Arabesque}, RStream \cite{RStream18}, G-Miner \cite{GMiner}, and Fractal \cite{Fractal}. SumPA \cite{SumPA} shows better performance than Peregrine and GraphPi. However, its code is not available. So, we run directed versions of same pattern graphs on same data sets reported in \cite{SumPA}. Our speedups (XMiner/Peregrine) are larger than SumPA's speedups (SumPA/Peregrine) even though SumPA is designed for undirected graph. Sandslash \cite{Sandslash} can not import graphs. According to the above reasons and the suggestions from authors of \cite{AutoMine}, we choose Peregrine and GraphPi as competitors. 

\newsavebox{\datasetbox}
\begin{lrbox}{\datasetbox}
\vspace{-10pt}
    \begin{tabular}{lrrcr}
        \hline
Datasets   & Nodes  & Arcs  & \begin{tabular}[c]{@{}l@{}}Average \\ Degree\end{tabular} & Size      \\
        \hline
Ego-Facebook\cite{Ego-Facebook} & 4,039    & 88,234      & 44         & 836KB    \\
Wiki-Vote\cite{Wiki-Vote}  & 7,115     & 103,689      & 29              & 968KB    \\
Cit-HepTh\cite{Cit-HepTh}  & 27,770    & 352,807     & 25              & 5.32MB    \\
Enron\cite{LWA}           & 69,244    & 276,143     & 8          & 3.13MB   \\
Web-BerkStan\cite{Web-Google}       & 685,230   & 7,600,595   & 22          & 105.03MB  \\
Web-Google\cite{Web-Google}     & 875,713   & 5,105,039   & 12         & 71.89MB  \\
Itwiki\cite{LWA}     & 1,016,867 & 25,619,926  & 50         & 339.49MB \\
Orkut\cite{Orkut}     & 3,072,441 & 117,185,083  & 76         & 1.64GB \\
Enwiki\cite{LWA}     & 6,047,510 & 142,691,609 & 47         & 2.09GB \\
Friendster\cite{Friendster}     & 65,608,366 & 1,806,067,135 & 55         & 30.14GB \\
\hline
\end{tabular}
\end{lrbox}
\begin{table}[htb]
\centering
\caption{Dataset characteristics}
\label{tab:dataset}
\scalebox{0.8}{\usebox{\datasetbox}}
\vspace{-10pt}
\end{table}\textbf{\textit{Datasets.}} In our experiments, we choose 10 datasets \cite{GraphPi,Peregrine} which have sizes ranging from a low of 88K arcs to a high of 1.8B arcs (Table \ref{tab:dataset}). The data sets are widely used for evaluation \cite{GraphPi,Peregrine}. \textit{Wiki-Vote} is a Wikipedia editorial voting graph. \textit{Cit-HepTh} is a paper citation network graph, and \textit{Enron} is a mail network graph. \textit{Web-Google} is a web link network.  \textit{Ego-Facebook, Itwiki, Orkut, Enwiki}, and \textit{Friendster} are social networks. Among them, although Peregrine reported its performance on \textit{Orkut} \cite{Peregrine}, we can not run it on \textit{Orkut} successfully.

\textbf{\textit{Patterns.}} We choose six pattern graphs which have vertex sizes ranging from 3 to 7 (Figure~\ref{fig:PatternGraph}). Each pattern graph is named according to its node size. For example, pattern graph 3\_p has three vertices.
The last two patterns have 7 vertices, so we name them as 7\_pb and 7\_pc respectively since the example pattern graph is 7\_pa. With the increase of the number of vertices and edges, the complexity of pattern graphs also increases rapidly. 
For fair comparison, the undirected versions of 6\_p, 7\_pb, and 7\_pc are from \cite{GraphPi}, the undirected versions of 4\_p, and 5\_p are similar to the pattern graphs used in the competitors \cite{GraphPi,Peregrine}. The undirected versions of 3\_p is a triangle counting task. 
Since the competitors are designed for undirected graph, the competitors will mine the un-directed versions.
\vspace{-1pt}
\begin{figure}[htbp]
\centering
\subfigure[3\_p]{
\includegraphics[width=2.3cm]{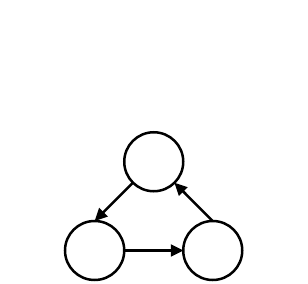}
\label{fig:3p}
}
\quad
\subfigure[4\_p]{
\includegraphics[width=2.3cm]{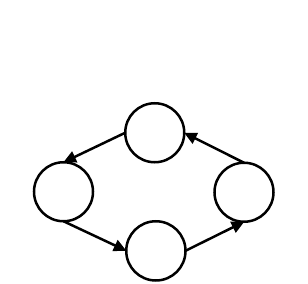}
\label{fig:4p}
}
\quad
\subfigure[5\_p]{
\includegraphics[width=2.3cm]{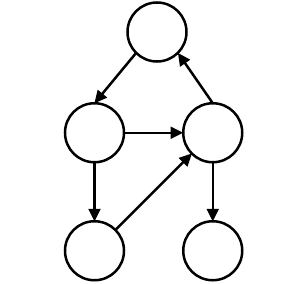}
\label{fig:5p}
}
\quad
\subfigure[6\_p]{
\includegraphics[trim=0 10 0 10,width=2.3cm]{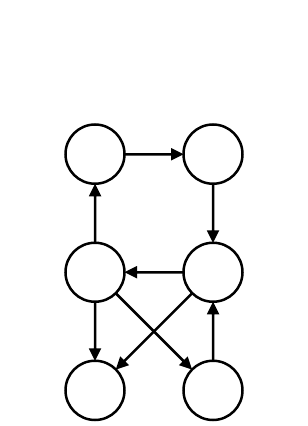}
\label{fig:6p}
}
\quad
\subfigure[7\_pb]{
\includegraphics[trim=0 10 0 10, width=2.3cm]{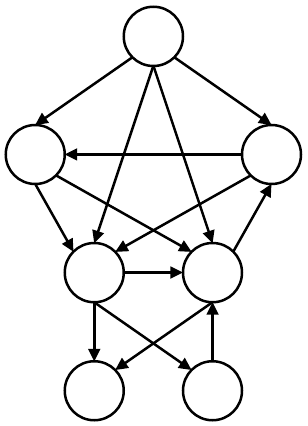}
\label{fig:8p}
}
\quad
\subfigure[7\_pc]{
\includegraphics[trim=0 10 0 10, width=2.3cm]{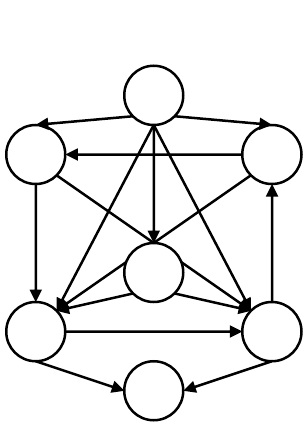}
\label{fig:7p}
}
\caption{Pattern Graphs}
\label{fig:PatternGraph}
\end{figure}

All the experiments run on a server with two Intel 14-core 2.0GHz CPUs, 256GB RAM, and 500GB disk.  The operating system is Ubuntu 18.04.5. We run each task on each dataset three times and average the execution times. The time reported in our evaluation does not include  loading time and pre-processing time of data sets.
\subsection{Performance Comparison}

\subsubsection{Data Graph}
We roughly divide data graphs into three groups: small graphs ($<$1 million arcs), medium graphs (1$\sim $100 million arcs), and large graphs ($>$100 million arcs). 
Table \ref{tab:perfsmall}, \ref{tab:perfmedium}, \ref{tab:perflarge} shows the performance.

\textbf{Small Graphs.} If a system runs a matching task for more than 5 hours on small and medium graphs, we will terminate the execution. XMiner and Peregrine can load four datasets while GraphPi can only process two small datasets. However, Peregrine need more than 5 hours to match 7\_pc on three datasets. For the  tasks less than 5 hours, XMiner outperforms GraphPi by 2.4 times (3\_p on Wiki-Vote)-165713 times (7\_pc on Ego-Facebook), and outperforms Peregrine by 1.8 times (3\_p on Wiki-Vote) - 185271 times (7\_pb on Ego-Facebook) in the running time. 
Peregrine slightly outperforms XMiner when matching 3\_p on Cit-HepTh. 
Although Ego-Facebook is the smallest graph, two competitors do not show better performance on Ego-Facebook. One reason is that its average degree is larger than other 5 data sets. Despite this, compared to the competitors, the performance of XMiner varies slightly.  
\newsavebox{\perfboxsmall}
\begin{lrbox}{\perfboxsmall}
\begin{tabular}{clrrr}
\hline
{Data Graph}   & {Pattern Graph} & XMiner & GraphPi & Peregrine  \\
\hline
        {Ego-Facebook}
&	3\_p	&  \textbf{0.006}  &	0.019 	&	0.011 	  \\
&	4\_p	&  \textbf{0.043}  &	0.137 	&	0.222 	  \\
&	5\_p	&  \textbf{0.060}  &	1.862 	&	23.379 	  \\
&	6\_p	&  \textbf{0.088}  &	153.692 	&	1089.760 	  \\
&	7\_pb	&  \textbf{0.086}  &	4910.986	&	15927.2	  \\
&	7\_pc	&  \textbf{0.058}  &	9554.014 	&	5 h+	  \\
        {Wiki-Vote}
&	3\_p	&  \textbf{0.007}  &	0.017 	&	0.013 	  \\
&	4\_p	&  \textbf{0.070}  &	0.412 	&	0.499 	  \\
&	5\_p	&  \textbf{0.486}  &	0.551 	&	9.900 	  \\
&	6\_p	&  \textbf{6.372}  &	112.789 	&	132.448 	  \\
&	7\_pb	&  \textbf{1.875}  &	56.473	&	84.954	  \\
&	7\_pc	&  \textbf{2.730}  &	47.444 	&	182.402 	  \\
        {Cit-HepTh}
&	3\_p	&  	0.035 	 & 	--	& \textbf{0.027} \\
&	4\_p	&  \textbf{0.099}  &	--	&	0.981 	  \\
&	5\_p	&  \textbf{0.178}  &	--	&	72.126 	  \\
&	6\_p	&  \textbf{0.278}  &	--	&	104.044 	  \\
&	7\_pb	&  \textbf{0.234}  &	--	&	177.177	  \\
&	7\_pc	&  \textbf{0.180}  &	--	&	5 h+	  \\
        {Enron}
&	3\_p	&  \textbf{0.015}  &	--	&	0.023 	  \\
&	4\_p	&  \textbf{0.251}  &	--	&	1.171 	  \\
&	5\_p	&  \textbf{3.417}  &	--	&	26.104 	  \\
&	6\_p	&  \textbf{27.917}  &	--	&	255.802 	  \\
&	7\_pb	&  \textbf{40.934}  &	--	&	657.993	  \\
&	7\_pc	&  \textbf{93.929}  &	--	&	5 h+	  \\
                                \hline
\end{tabular}
\end{lrbox}
\begin{table}[htbp]
\centering
\caption{Performance on Small Graphs (in seconds)}
\label{tab:perfsmall}
\scalebox{0.8}{\usebox{\perfboxsmall}}
\parbox{14.1cm}{\footnotesize{-- indicates the system can not load the graph.} }
\end{table}

\textbf{Medium Graphs.} 
It takes more than 5 hours for the competitors to run the complex tasks (e.g. GraphPi matches 6\_p, 7\_pb, and 7\_pc on Web-BerkStan). Three systems can successfully process all pattern graphs on Web-Google except that Peregrine runs timeout on 7\_pc. For the tasks which are not out of time, XMiner outperforms GraphPi by 2.3 times (5\_p on Web-Google)-69.7 times (5\_p on Web-BerkStan), and outperforms Peregrine by 2.2 times (3\_p on Web-Google)) - 269.3 times (5\_p on Web-Google) in the running time.
\vspace{-2pt}
\newsavebox{\perfboxmid}
\begin{lrbox}{\perfboxmid}
\begin{tabular}{clrrr}
\hline
{Data Graph}   & {Pattern Graph} & XMiner & GraphPi & Peregrine  \\
\hline
        {Web-BerkStan}
&	3\_p	&  \textbf{0.082}  &	1.310 	&	--	  \\
&	4\_p	&  \textbf{4.454}  &	39.057 	&	--	  \\
&	5\_p	&  \textbf{19.767}  &	1377.565 	&	--	  \\
&	6\_p	&  \textbf{72.531}  &	5 h+	&	--	  \\
&	7\_pb	&  \textbf{116.148}  &	5 h+	&	--	  \\
&	7\_pc	&  \textbf{197.783}  &	5 h+	&	--	  \\
        {Web-Google}
&	3\_p	&  \textbf{0.072}  &	0.361 	&	0.161 	  \\
&	4\_p	&  \textbf{0.213}  &	0.691 	&	14.781 	  \\
&	5\_p	&  \textbf{2.189}  &	4.990 	&	589.590 	  \\
&	6\_p	&  \textbf{19.945}  &	1385.263 	&	1574.880 	  \\
&	7\_pb	&  \textbf{66.629}  &	268.428	&	2116.41	  \\
&	7\_pc	&  \textbf{120.553}  &	361.278 	&	5 h+	  \\
        {Itwiki}
&	3\_p	&  \textbf{0.958}  &	--	&	3.837 	  \\
&	4\_p	&  \textbf{148.549}  &	--	&	3304.360 	  \\
&	5\_p	&  \textbf{196.193}  &	--	&	5 h+	  \\
&	6\_p	&  \textbf{2487.970}  &	--	&	5 h+	  \\
&	7\_pb	&  \textbf{426.024}  &	--	&	5 h+	  \\
&	7\_pc	&  \textbf{1272.540}  &	--	&	5 h+	  \\
                    \hline
\end{tabular}
\end{lrbox}
\begin{table}[ht]
\centering
\caption{Performance on Medium  Graphs (in seconds)}
\label{tab:perfmedium}
\scalebox{0.8}{\usebox{\perfboxmid}}
\vspace{-2pt}
\end{table}

\textbf{Large Graphs.} XMiner also shows the best performance on large graphs (Table \ref{tab:perflarge}). For the large graphs, the average degrees are 47-76 which is much larger than most of small and medium graphs. Their maximal degrees are also very large. For example, Enwiki's maximum out-degree and maximum in-degree are 236,348 and 10,694,  respectively.  So, it is difficult for graph matching systems to process them because they will incur large amount of search and huge intermediate results. GraphPi can process 2 large graphs, and Peregrine can only process one large input graph (Table \ref{tab:perflarge}). For large graphs, XMiner is 4.5 times and 3.3 times better than GraphPi and Peregrine respectively when matching small pattern graphs. On Orkut, GraphPi takes more than 24 hours to process complex patterns, while XMiner only needs less than 116 seconds for the tasks. The reason is that XMiner employs the pattern graph reduction method, and designs the corresponding  execution mechanism. So it can process data graphs of different sizes and complete the matching tasks well. Compared with other systems, XMiner supports processing complex pattern graphs on larger graphs, thus has good applicability. 
\vspace{-1pt}
\newsavebox{\perfboxlarge}
\begin{lrbox}{\perfboxlarge}
\begin{tabular}{clrrr}
\hline
{Data Graph}   & {Pattern Graph} & XMiner & GraphPi & Peregrine  \\
\hline
        {Orkut}
&	3\_p	&  \textbf{2.896}  &	12.999 	&	--	  \\
&	4\_p	&  \textbf{559.143}  &	1146.607 	&	--	  \\
&	5\_p	&  \textbf{50.835}  &	657.341 	&	--	  \\
&	6\_p	&  \textbf{100.167}  &	24 h+	&	--	  \\
&	7\_pb	&  \textbf{115.478}  &	24 h+	&	--	  \\
&	7\_pc	&  \textbf{69.503}  &	24 h+	&	--	  \\
        {Enwiki}
&	3\_p	&  \textbf{6.153}  &	--	&	--	  \\
&	4\_p	&  \textbf{776.628}  &	--	&	--	  \\
&	5\_p	&  \textbf{1031.360}  &	--	&	--	  \\
&	6\_p	&  \textbf{34228.100}  &	--	&	--	  \\
&	7\_pb	&  \textbf{30818.100}  &	--	&	--	  \\
&	7\_pc	&  \textbf{55984.300}  &	--	&	--	  \\

         {Friendster}
&	3\_p	&  \textbf{80.246}  &	258.785  	&	194.253  	  \\
&	4\_p	&  \textbf{11567.200}  & 32480.943 	&	67012.800  	  \\
&	5\_p	&  \textbf{880.359}  &	2323.516	&	89425.3	  \\
&	6\_p	&  \textbf{2147.33}  &	24h+	&	24h+	  \\
&	7\_pb	&  \textbf{1256.98}  &	24h+	&	24h+	  \\
&	7\_pc	&  \textbf{948.355}  &	24h+	&	24h+	  \\
\hline

\end{tabular}
\end{lrbox}
\begin{table}[htbp]
\centering
\caption{Performance on Large Graphs (in seconds)}
\label{tab:perflarge}
\scalebox{0.8}{\usebox{\perfboxlarge}}
\end{table}

On average, XMiner outperforms GraphPi  by 8329.7$\times$ and Peregrine by 6677$\times$,  respectively (success tasks). 
\subsubsection{Pattern Graphs}
We refer Figure \ref{fig:3p}, \ref{fig:4p}, \ref{fig:5p} as small pattern graphs. 
and pattern graphs (Figure \ref{fig:6p}-\ref{fig:7p}) 
as complex pattern graphs.

\textbf{Small pattern graphs.}  Undirected version of 3\_P and 4\_p are common patterns widely used in evaluation. 
XMiner outperforms both GraphPi and Peregrine by a large margin in most cases except that XMiner is slightly slower than Peregrine when matching 3\_p on Cit-HepTh.  When matching 5\_P, XMiner is much faster than the two competitors, especially the large graph. One reason is that XMiner reduces the patterns. Since small pattern contains few vertices and edges, XMiner can not remove many edges as it does on complex patterns. So, XMiner does not improve performance much by finding constraint relationships.

\textbf{Complex pattern graphs.}  
As the complexity of pattern graphs increases, search spaces are becoming larger, and thus the sizes of intermediate results grow rapidly as well. 
Three systems will require more time to find all the embeddings of the patterns.
The competitors time out when they run some complex matching tasks. For example, both Peregrine and GraphPi time out when they process pattern graph 6\_p, 7\_pb, and 7\_pc on large data graphs although the undirected versions of the pattern graphs have symmetry structures. When matching the complex pattern graphs, XMiner has significant advantages over both Peregrine and GraphPi. Specifically, XMiner is much faster than Peregrine by 185271$\times$ when matching 7\_pb on Ego-Facebook and faster than GraphPi by up to 6 orders of magnitude (165712.9$\times$) when matching 7\_pb on Ego-Facebook. When processing 6\_p, 7\_pb, and 7\_pc, Peregrine and GraphPi require more than 24 hours on large graphs except those unsuccessful tasks. On the contrary, XMiner finishes the tasks matching 7\_pc in 0.058  (Ego-Facebook) - 55984.3 (Enwiki) seconds. The reason is that XMiner explores data graphs using reduced pattern graphs and other edges included by reduced graphs can be materialized using  intermediate results. 
\subsection{Performance Impact of Pattern Reduction}
In order to  evaluate pattern reduction effect on performance, we compare single threaded XMiner with undirected graph matching algorithm QGLfs and RIfs \cite{SunSIGMOD20} whose codes are available. QGLfs and RIfs are single-threaded program. We choose three data sets: \textit{Ego-Facebook}, \textit{enron}, \textit{web-Google}, and \textit{Orkut}, which are from small, medium, and large datasets, respectively. QGLfs and RIfs time out or fail when they handle complex patterns or large data sets. For example, it take QGLfs and RIfs more than 5 hours to run 4\_p on \textit{Orkut}. XMiner can process 6 tasks on three datasets. For the successful tasks, the speedups (competitor/XMiner) are between 1.03 (3\_p on \textit{enron})-8.36 (5\_p on \textit{Orkut}). 

We also implement a simplified XMiner without pattern reduction. The simplified XMiner generates a plan for each original pattern graph and executes each plan as XMiner does.  We run four tasks (5\_p, 6\_p, 7\_pb, and 7\_pc) on dataset \textit{enron}, \textit{web-Google}, and \textit{Orkut} because  complex tasks can show the effect of pattern reduction clearly. The results (Figure \ref{fig:timeworeduction}) show XMiner is much faster than XMiner without pattern reduction by 2.2 (6\_p on \textit{web-Google})-32.21 (5\_p on \textit{Orkut}). 
\begin{figure}[htb]
    \centering
    \includegraphics[trim=0 20 0 10, scale=0.27]{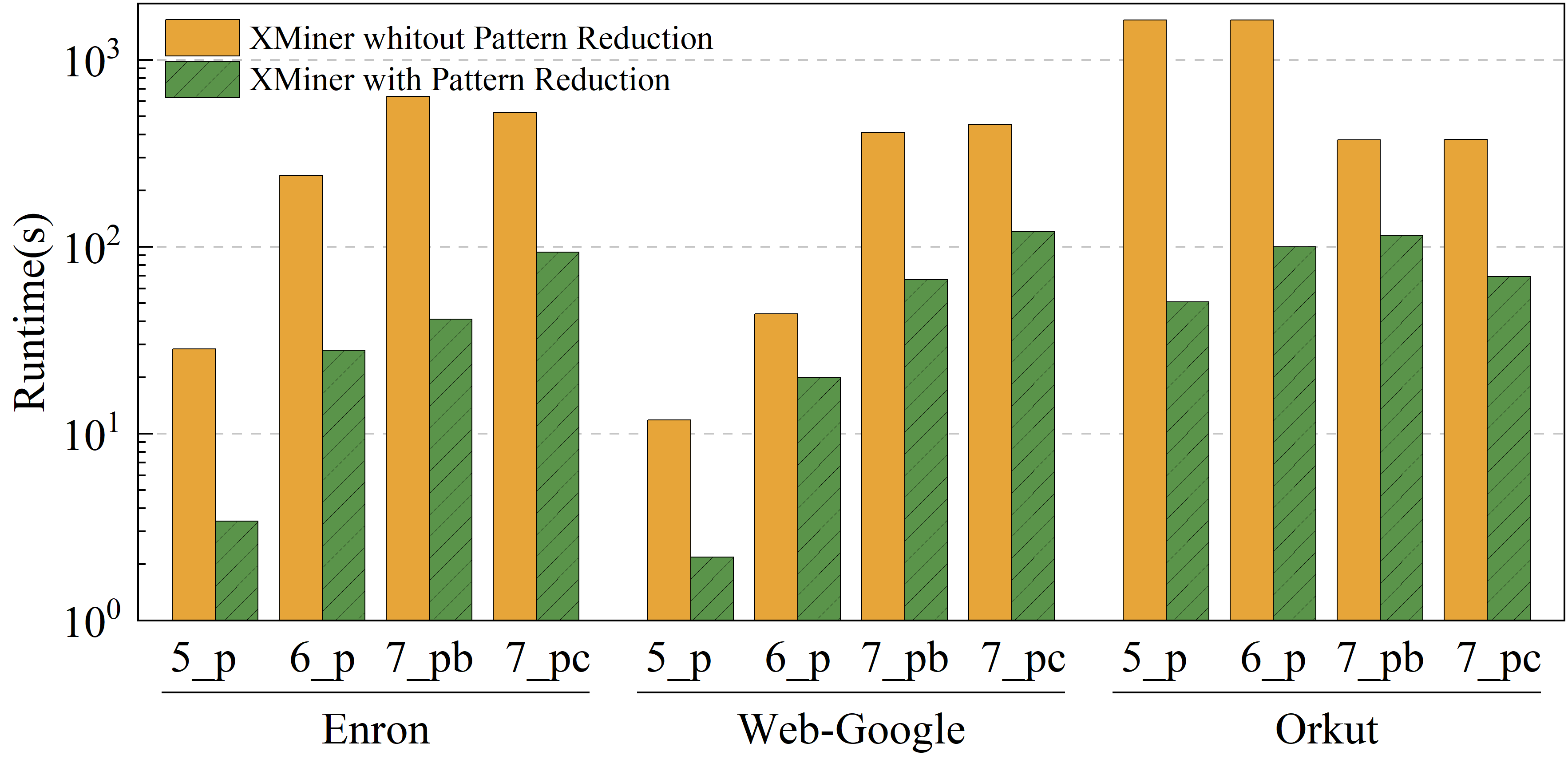}
    \caption{XMiner with/without Pattern Reduction}
    \label{fig:timeworeduction}
\end{figure}
\begin{figure}[htb]
\centering
\includegraphics[trim=0 20 0 10,scale=0.27]{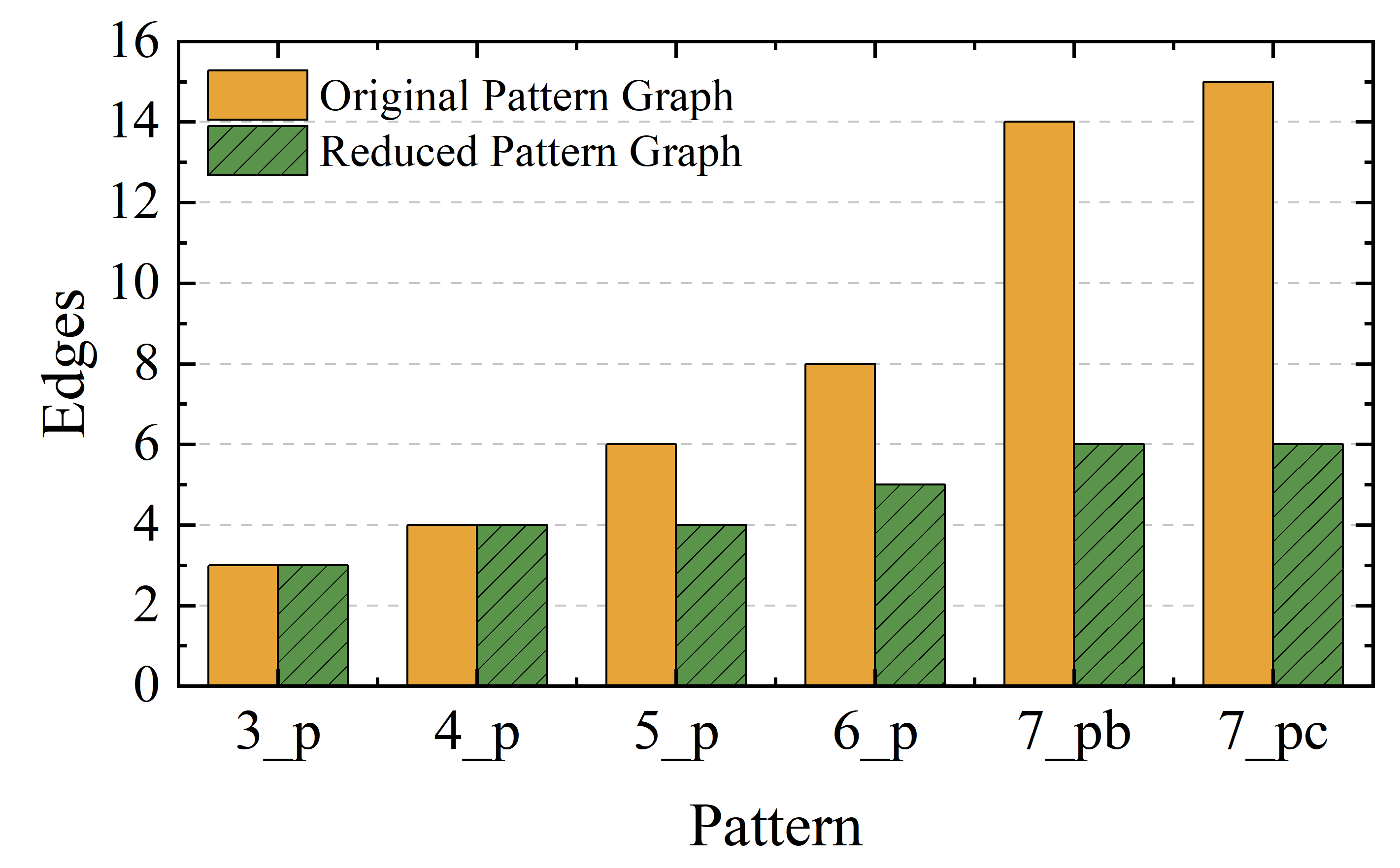}
\caption{Pattern Graphs after Reduction}
\label{fig:reducedpatternsize}
\end{figure}

We records the number of edges of reduced pattern graphs (Figure \ref{fig:reducedpatternsize}). The edge numbers of reduced pattern graphs are between 3 and 6. 
Compared to the original pattern graphs, the more complex a pattern graph is, the more edges it will be temporarily reduced. For example, reduced 7\_pc only has 6 edges while the original 7\_pc has 15 edges. Reduced 7\_pb also has 6 edges left. However, both 3\_P and 4\_P are loops in which their vertices and edges are equivalent, respectively. So, XMiner does not reduce them.
\subsection{Scalability}
\label{sec:scal}
We vary the thread number to evaluate its scalability under different threads. The number of threads is set to be 1, 2, 4, 8, 16, and 28, respectively because the server has 28 physical cores. We choose to mine 4\_p on data set \textit{Web-Google} because our pattern reduction approach can not reduce many intermediate results comparing to the symmetry breaking used in the two competitors. And the competitors can also process it. The speedups, which are the ratio of the running time by one thread to that by the given number of threads, are shown in Fig.~\ref{fig:Scalability3}. In general, the speedup increases as the number of
threads increases. 
GraphPi performs not better than other two systems. When the number of threads increases from 1 to 8, the XMiner's speedups are slightly better than Peregrine's speedups. When the number of threads grows from 8 to 28, the speedups of XMiner increase from 7.634 to 24.339, while Peregrine increased from 7.550 to 23.089. 
\begin{figure}[htbp]
    \centering
    \includegraphics[trim=0 20 0 0, clip, scale=0.26]{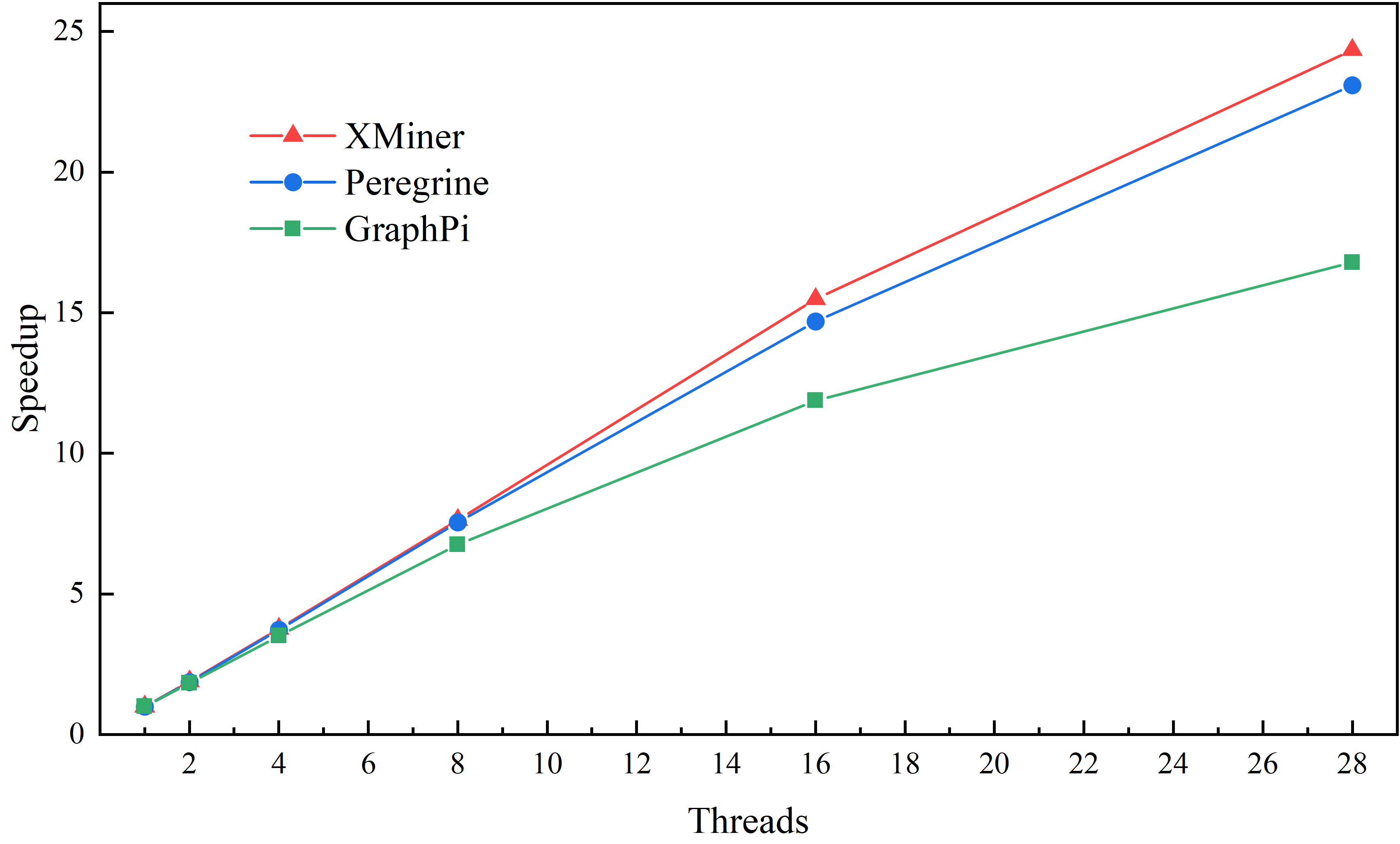}
    \caption{Speedup}
    \label{fig:Scalability3}
\end{figure}
\vspace{-2pt}

XMiner can scale to large search space by pattern reduction. To demonstrate this scalability, we make XMiner to start its exploration from the given percentage of start nodes. If we randomly extract the given sizes of data from Orkut, some results may be not valid again because of incompleteness and less connectivity in the generated datasets.
Here, we randomly choose the given percents of start nodes matching the start vertices in pattern graphs. Then we search Orkut through BFS from those start nodes until BFS can not continue. The returned subgraphs form a data set.

By this way , search space will be same as original. We generate ten sets of exploration tasks which contain 10\%-100\% of start nodes, respectively for each pattern graph. On these exploration task sets, XMiner mines pattern 5\_p-7\_pc, respectively. The two competitors have their own exploration strategies and may start the explorations from other vertices. So, here we do not compare with them. The execution times are shown in Figure \ref{fig:Scalability2}. With the increase of the percentage, the running time is about linear growth, except for 20\% dataset. The reason is that the number of matchings of subgraphs originated from each start node varies. So, the size of matchings of datasets is not in line with the percentage of start nodes. Its running time changes sharply when the percentage is low. When the percentage grows, the curve is smoother.
\begin{figure}[htb]
    \centering
    \includegraphics[trim=0 20 0 0, clip, scale=0.27]{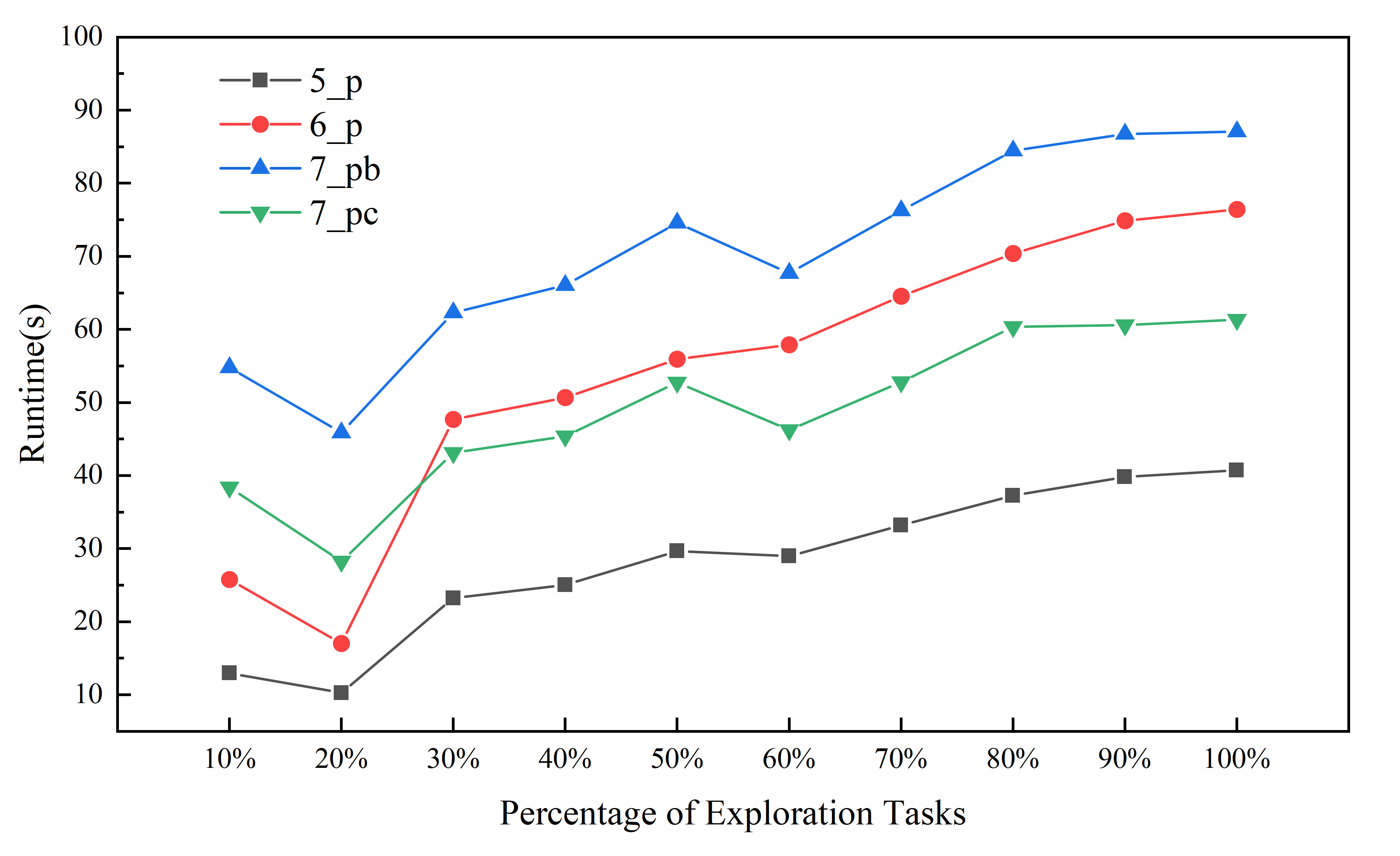}
    \caption{Varying the Percentage of the Exploration Tasks}
    \label{fig:Scalability2}
\end{figure}
\vspace{-10pt}
\subsection{Memory Consumption}
The graph matching systems tend to generate a large-scale intermediate result sets. Storing the results will require huge memory space, and worse the performance. XMiner can not only accelerate the graph matching task, but also require less memory in the process of graph matching. Here, 
we run the matching tasks on the data sets mentioned in Section \ref{sec:scal} and record the peak memory consumption. We exclude those memory for the data graph since data graph is resident in memory. As shown in Figure \ref{fig:memoryconsumption},  when XMiner performs matching tasks of different pattern graphs, the peak memory consumption increases smoothly 
with the growing percentage, and there is no sharp increase in memory consumption when it performs more matching tasks (e.g. 80\%-100\%). One reason is that XMiner need not maintain many redundant data, but reuse intermediate results for edges and vertices included by others.

\begin{figure}[htb]
\centering
\includegraphics[trim=0 20 0 0, clip,scale=0.27]{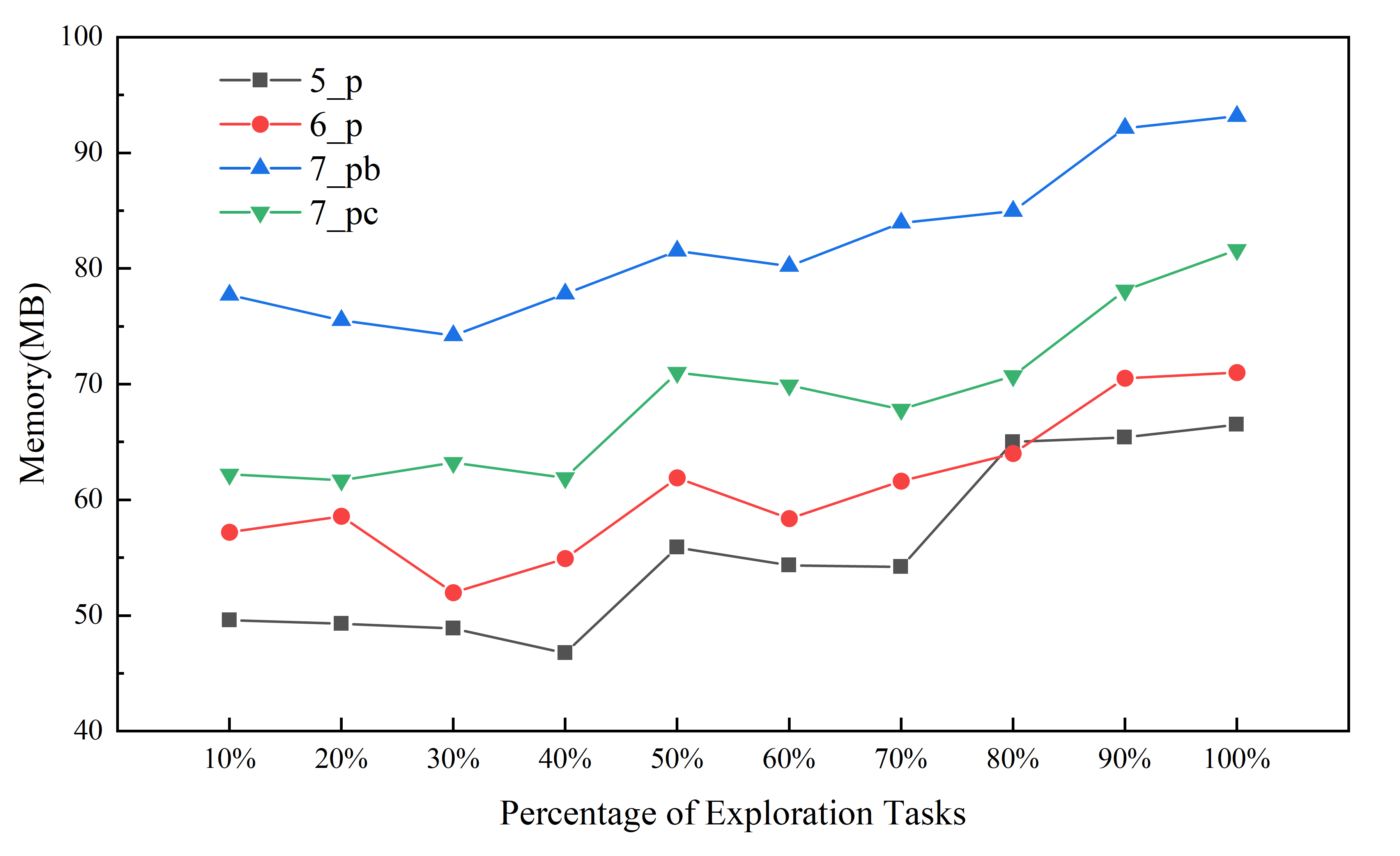}
\caption{Memory Consumption}
\label{fig:memoryconsumption}
\end{figure}
We also observe that the peak memory for intermediate results grows when graph patterns are becoming more complex. The reason is that matching complex pattern graphs need more explorations and thus more intermediate results are resident in memory. However, it does not grows rapidly because XMiner reduces pattern graph. For example, After reduction, there are 4 edges left in 5\_p and 6 edges left in 7\_pb and 7\_pc for exploration. The partial matchings of these edges will be reused for initializing other edges. So, when matching 7\_pb, the peak memory consumption reaches about 95MB while it requires at most 65MB memory for holding the partial matchings of 5\_p during the matching process. 
\section{Related Work}
\label{sec:related}
%
Graph mining receives constant attention and many graph mining approaches are reported which  can be roughly divided into three classes:  stand-alone graph mining, distributed graph mining \cite{QFrag,NScale, Arabesque,Fractal, DistGraph, PruneJuice18, GMiner, G-thinker}, and hardware acceleration of graph mining\cite{PBE,GSI,PBE2,GAMMA,G2Miner, Gramer,Tesseract,FlexMiner,FINGERS, IntersectX,SISA,NDMiner}. Here, we focus on stand-alone graph matching. 

Many graph matching algorithms are proposed \cite{SunSIGMOD20,Kim2021}. Since most of them are single-threaded program, they can not process complex patterns or large graphs in reasonable time. 
So, many algorithms or systems explore how to mine subgraphs in parallel. 
The systems \cite{RStream18,Mhedhbi19,Han19,Pangolin,Sandslash,AutoMine} explore programming model or execution model so that users can easily and efficiently mine specified graph patterns using them. 
Some disk-based systems, such as OPT \cite{OPT}, DualSim \cite{DUALSIM16}, Kaleido \cite{Kaleido} overlapp computation and IO operations since is much slower than CPU. 


Pattern graphs imply some special sub-structures, such as symmetry. So, there exist a lot of redundant computations and data access. To avoid redundant computation, GraphZero \cite{GraphZero} tries to destruct symmetry of pattern graph according to group theory. Similarly, GraphPi \cite{GraphPi} utilizes 2-cycles of group theory to generate multiple sets of asymmetric restrictions, so that it can eliminate automorphisms.
Peregrine \cite{Peregrine} analyzes the pattern in advance to eliminate its symmetry, avoids redundant isomorphism detections.
Fractal \cite{Fractal} also uses symmetry breaking for pattern matching. 
Similarly, AutoMine \cite{AutoMine} also employs symmetry breaking for some use cases. It may incur duplicate matches. So users need examine each single match to filter duplicate matches. VEQs \cite{Kim2021} identify equivalences of vertices in order to reduce
search space during subgraph search. SumPA \cite{SumPA} explores intra-pattern and inter-pattern computational redundancies to improve the mining performance. The work of the above systems is mainly focused on the optimization of isomorphism detection and redundant computing. There exist many set operations on intermediate results in a graph mining task. Ordering the set operations can eliminate redundancies in some sense. So, Cyclosa \cite{Cyclosa}, an undirected graph pattern mining framework transfers the set operations for a pattern into a set dataflow and then executes the dataflow in parallel. It reduces the redundant computation. Lack of fine-grained inside analysis of patterns not only makes performance improvement of these systems more difficult, but also limits their applicability to more complex matching tasks. 
XMiner also uses symmetry destruction and the generation of execution plans. At the same time,
XMiner employs a pattern reduction method, which can effectively simplify complex pattern graphs, reduce the internal overhead in the matching process, and improve the performance.

$Turbo_{iso}$ \cite{TurboISO}. an approach for undirected subgraph
search proposes concept \textit{Neighborhood Equivalence Class} (NEC), in which all vertices have same embeddings. $Turbo_{iso}$ rewrites a pattern graph into an NEC tree, where each vertex corresponds to an NEC. For complex pattern graphs, many vertices are not in an NEC. $Turbo_{iso}$ only generates combinations for each NEC. If a combination does not contribute to a solution, all possible permutations for the combination will be pruned without any enumeration. We would like to note that, our constraint inclusion set is essentially different from the NEC concept, in which all vertices have the same label and the same set of adjacent query vertices. Moreover, \textit{NEC} does not involve edges.  
Furthermore, XMiner employs a pattern reduction method 
and thus reduces the internal overhead in the matching process.
\section{Conclusions}
\label{sec:conclusion}
We have presented XMiner based on a pattern reduction approach for fast and efficient digraph matching. 
XMiner is flexible enough to work on large range of graphs and pattern graphs. There are several important features of XMiner. First, XMiner is characterized as a reduction approach. So, XMiner has no trouble handling complex pattern graphs that other systems can not accomplish. Second, our definition of graph matching problem focuses on analysis of the relationship between constraints. So, it is flexible enough to capture the characteristics of graph matching problems. Third, identifying constraint-inclusion relationship helps large number of redundant operations, such as data access and computations etc, especially when pattern graph is complex and data graph is very large. We have demonstrated that the XMiner outperforms GraphPi and Peregrine on real graphs of varying sizes for a variety of pattern graphs although the competitors are designed for undirected graphs. The results also show XMiner is applicable and scalable.
\small
\bibliographystyle{ACM-Reference-Format}

\balance
\end{document}